%% file: paper_main_ICALP.tex
\documentclass[a4paper,UKenglish,cleveref, autoref, thm-restate]{lipics-v2021}

\usepackage{amssymb, amsthm, amsmath, mathtools}
\usepackage{calrsfs}
\usepackage[linesnumbered,lined,ruled, commentsnumbered, noend]{algorithm2e}
\usepackage{multirow}
\usepackage{multicol}
\usepackage{graphicx}
\usepackage{aliascnt}
\usepackage{balance}
\usepackage{mathtools}
\usepackage[normalem]{ulem}
\usepackage{booktabs}
\usepackage{fancyvrb}
\usepackage{bm}
\usepackage{colortbl}
\usepackage{footnote}
\usepackage{xspace}
\usepackage{url}
\usepackage{textcomp}
\usepackage{booktabs}
\usepackage{hyperref}
\usepackage{color}
\usepackage{xspace}
\usepackage{caption}
\renewcommand{\paragraph}[1]{\medskip \noindent {\bf #1}}
\newcommand{\OUT}{{\text{out}}}
\newcommand{\D}{\mathcal{D}}
\newcommand{\poly}{\operatorname{poly}}

\newcommand{\diam}{\operatorname{diam}}
\newcommand{\W}{\mathcal{W}}

\newcommand{\C}{\mathcal{C}}
\renewcommand{\D}{\mathcal{O}}
\newcommand{\Q}{\mathcal{Q}}
\newcommand{\dist}{\phi}

\newcommand{\eps}{\varepsilon}
\newcommand{\T}{\mathcal{T}}

\newcommand{\norm}[1]{\left\lVert#1\right\rVert}
\renewcommand{\O}{\widetilde{O}}

\newcommand{\false}{\textrm{false}}

\renewcommand{\Re}{\mathbb{R}}
\newcommand{\graph}{G'}
\newcommand{\setCliques}{\mathcal{F}}
\renewcommand{\H}{\mathcal{H}}
\newcommand{\AM}{\bar{A}_\square}
\newcommand{\BM}{\bar{B}_\square}

\title{Dynamic Enumeration of Similarity Joins} 

\author{Pankaj K. Agarwal}{Duke University}{}{}{}

\author{Xiao Hu}{Duke University}{}{}{}

\author{Stavros Sintos}{University of Chicago}{}{}{}

\author{Jun Yang}{Duke University}{}{}{}

\authorrunning{P. K. Agarwal et al.} 

\Copyright{Pankaj K. Agarwal, Xiao Hu, Stavros Sintos, and Jun Yang} 

\ccsdesc[100]{Theory of computation $\rightarrow$ Theory and algorithms for application domains $\rightarrow$ Database theory $\rightarrow$ Data structures and algorithms for data management
} 

\keywords{dynamic enumeration, similarity joins, worst-case delay guarantee} 

\funding{This work has partially supported by NSF grants IIS-1814493 and CCF-2007556.}

\nolinenumbers 

\EventEditors{Nikhil Bansal, Emanuela Merelli, and James Worrell}
\EventNoEds{3}
\EventLongTitle{48th International Colloquium on Automata, Languages, and Programming (ICALP 2021)}
\EventShortTitle{ICALP 2021}
\EventAcronym{ICALP}
\EventYear{2021}
\EventDate{July 12--16, 2021}
\EventLocation{Glasgow, Scotland (Virtual Conference)}
\EventLogo{}
\SeriesVolume{198}
\ArticleNo{5}

\DeclareMathOperator{\sign}{sgn}

\begin{document}

\maketitle

\begin{abstract}
This paper considers enumerating answers to {\em similarity-join} queries under dynamic updates: Given two sets of $n$ points $A,B$ in $\Re^d$, a metric $\dist(\cdot)$, and a distance threshold $r > 0$, report all pairs of points $(a, b) \in A \times B$ with $\dist(a,b) \le r$. Our goal is to store $A,B$ into a dynamic data structure that, whenever asked, can enumerate all result pairs with worst-case {\em delay} guarantee, i.e., the time between enumerating two consecutive pairs is bounded. Furthermore, the data structure can be efficiently updated when a point is inserted into or deleted from $A$ or $B$. 

We propose several efficient data structures for answering similarity-join queries in low dimension. For exact enumeration of similarity join, we present near-linear-size data structures for $\ell_1, \ell_\infty$ metrics with $\log^{O(1)} n$ update time and delay. We show that such a data structure is not feasible for the $\ell_2$ metric for $d \ge 4$. For {\em approximate} enumeration of similarity join, where the distance threshold is a soft constraint, we obtain a unified  linear-size data structure for $\ell_p$ metric, with $\log^{O(1)} n$ delay and update time. In high dimensions, we present an efficient data structure with worst-case delay-guarantee using {\em locality sensitive hashing} (LSH).  
\end{abstract}

\input{introduction}
\input{Linfty}
\input{L2}
\input{lshMain}

\bibliographystyle{abbrv}
\bibliography{paper_main_ICALP}

\clearpage

\input{appendix}
\end{document}

%% file: introduction.tex
\section{Introduction}
	\label{sec:intro}
	There has been extensive work in many areas including theoretical computer science, computational geometry, and database systems on designing efficient dynamic data structures to store a set $\D$ of objects so that certain queries on $\D$ can be answered quickly and objects can be inserted into or deleted from $\D$ dynamically. A query $\Q$ is specified by a set of constraints and the goal is to report the subset $\Q(\D) \subseteq \D$ of objects that satisfy the constraints, the so-called {\em reporting} or {\em enumeration} queries. More generally, $\Q$ may be specified on $k$-tuples of objects in $\D$, and we return the subset of $\D^k$ that satisfy $\Q$. One may also ask to return certain statistics on $\Q(\D)$ instead of $\Q(\D)$ itself, but here we focus on enumeration queries. As an example, $\D$ is set of points in $\Re^d$ and a query $\Q$ specifies a simple geometric region $\Delta$ (e.g., box, ball, simplex) and asks to return $\D \cap \Delta$, the so-called {\em range-reporting} problem. As another example, $\D$ is again a set of points in $\Re^d$, and $\Q$ now specifies a value $r \ge 0$ and asks to return all pairs $(p,q) \in \D \times \D$ with $\lVert p-q\rVert \le r$. Traditionally, the performance of a data structure has been measured by its size, the time needed to update the data structure when an object is inserted or deleted, and the {\em total time} spent in reporting $\Q(\D)$. In some applications, especially in exploratory or interactive data analysis of large datasets, it is desirable to report $\Q(\D)$ incrementally one by one so that users 
    can start exploiting the first answers while waiting for the remaining ones.  To offer guarantees on the regularity during the enumeration process, maximum {\em delay} between the enumeration of two consecutive objects has emerged as an important complexity measure of a data structure~\cite{bagan2007acyclic}. Formally speaking, {\em $\delta$-delay enumeration} requires that the time between the start of the enumeration process to the first result, the time between any consecutive pair of results, and the time between the last result and the termination of the enumeration process should be at most $\delta$.

	In this paper, we are interested in dynamic data structures for {\em (binary) similarity join} queries, which have numerous applications in data cleaning, data integration, collaborative filtering, etc. Given two sets of points $A$ and $B$ in $\Re^d$, a metric $\dist(\cdot)$, and a distance threshold $r > 0$, the similarity join asks to report  
	all pairs of $(a,b) \in A \times B$ with $\dist(a,b) \le r$.
	Similarity joins have been extensively studied in the database and data mining literature~\cite{chaudhuri2006primitive,jacox2008metric,paredes2009solving,silva2010similarity, wang2012can}, but it is still unclear how to enumerate similarity join results efficiently when the underlying data is updated. 
	%
	Our goal is to design a dynamic data structure that can be efficiently updated when an input point is inserted or deleted; and whenever an enumeration query is issued,  all join results can be enumerated from it with {\em worst-case delay} guarantee. 

\subsection{Previous results}
    
    We briefly review the previous work on similarity join and related problems. See surveys~\cite{al2016survey, augsten2013similarity, silva2016experimental} for more results.
    
    \paragraph{Enumeration of Conjunctive Query.} Conjunctive queries are built upon natural join ($\Join$), which is a special case of similarity join with $r = 0$, i.e., two tuples can be joined if and only if they have the same value on the join attributes.
    Enumeration of conjunctive queries has been extensively studied in the static settings~\cite{bagan2007acyclic,segoufin2013enumerating, carmeli2019enumeration} for a long time. 
    In 2017, two simultaneous papers~\cite{berkholz2017answering, idris2017dynamic} started to study dynamic enumeration of conjunctive query. Both obtained a dichotomy that a linear-size data structure that can be updated in $O(1)$ time while supporting $O(1)$-delay enumeration, exists for a conjunctive query if and only if it is {\em q-hierarchical} (e.g., the degenerated natural join over two tables is q-hierarchical). 
    However, for non-q-hierarchical queries with input size $n$, they showed a lower bound  $\Omega(n^{\frac{1}{2} - \eps})$ on the update time for any small constant $\eps > 0$, if aiming at $O(1)$ delay. This result is very negative since q-hierarchical queries are a very restricted class; for example, the matrix multiplication query $\pi_{X,Z} R_1(X,Y) \Join R_2(Y,Z)$, where $\pi_{X,Y}$ denotes the projection on attributes $X, Y$, and the triangle join $R_1(X,Y) \Join R_2(Y,Z) \Join R_3(Z,X)$ are already non-q-hierarchical.
    Later, Kara et al.~\cite{kara2019counting} designed optimal data structures supporting $O(\sqrt{n})$-time maintenance for some selected non-q-hierarchical queries like the triangle query etc.
    However, it is still unclear if a data structure of $O(\sqrt{n})$-time maintenance can be obtained for a large class of queries. Some additional trade-off results have been obtained in~\cite{kara2020trade, wang2020maintaining}.

	\paragraph{Range search.}  A widely studied problem related to similarity join is {\em range searching}~\cite{agarwal2017simplex, agarwal1999geometric, bentley1979data, willard1996applications}:
	Preprocess a set $A$ of points in $\Re^d$ with a data structure so that for a query range $\gamma$ (e.g., rectangle, ball, simplex), all points of $A \cap \gamma$ can be reported quickly. 
	A particular instance of range searching, the so-called {\em fixed-radius-neighbor} searching, in which the range is a ball of fixed radius centered at query point is particularly relevant for similarity joins. For a given metric $\phi$, let $\mathcal{B}_\phi(x,r)$ be the ball of radius $r$ centered at $x$.
	A similarity join between two sets $A, B$ can be answered by querying $A$ with ranges $\mathcal{B}_\phi(b,r)$ for all $b \in B$.
	%
	%
	Notwithstanding this close relationship between range searching and similarity join, the data structures for the former cannot be used for the latter: It is too expensive to query $A$ with $\mathcal{B}_\phi(b,r)$ for every $b \in B$ whenever an enumeration query is issued, especially since many such range queries may return empty set, and it is not clear how to maintain the query results as the input set $A$ changes dynamically. 
	

\paragraph{Reporting neighbors.}
    The problem of reporting neighbors is identical to our problem in the offline setting. In particular, given a set  $P$ of $n$ points in $\Re^d$ and a parameter $r$, the goal is to report all pairs of $P$ within distance $r$. The algorithm proposed in~\cite{lenhof1995sequential} can be modified to solve the problem of reporting neighbors under the $\ell_\infty$ metric in $O(n+k)$ time, where $k$ is the output size. Aiger et al.~\cite{aiger2014reporting} proposed randomized algorithms for reporting neighbors using the $\ell_2$ metric in $O((n+k)\log n)$ time, for constant $d$.
    
\paragraph{Scalable continuous query processing.}
    There has been some work on scalable {\em continuous query} processing, especially in the context of data streams~\cite{chen2000niagaracq, chandrasekaran2002streaming, wu2004interval} and publish/subscribe~\cite{fabret2001filtering}, where the queries are standing queries and whenever a new data item arrives, the goal is to report all queries that are affected by the new item~\cite{agarwal2006scalable, agarwal2005monitoring}. 
    In the context of similarity join, one can view $A$ as the data stream and $\mathcal{B}_\phi(b,r)$ as standing queries, and we update the results of queries as new points in $A$ arrive. There are, however, significant differences with similarity joins---arbitrary deletions
	are not handled; continuous queries do not need to return previously produced results; basing enumeration queries on a solution for continuous queries would require accessing previous results, which can be prohibitive if stored explicitly.

\subsection{Our results}

\begin{table}[t]
    \centering
    \begin{tabular}{|c|c|c|c|c|c|}
    \hline
       \multirow{2}{*}{Enumeration} &  \multirow{2}{*}{Metric} &  \multirow{2}{*}{Properties} & \multicolumn{3}{c|}{Data Structures} \\
        \cline{4-6}
        & & & Space & Update & Delay \\ 
       \hline
    \multirow{2}{*}{Exact} & $\ell_1/\ell_\infty$ & $r$ is fixed  & $\O(n)$ & $\O(1)$& $\O(1)$ \\
    \cline{2-6}
     & $\ell_2$
    & $r$ is fixed & $\O(n)$ & $\O(n^{1-\frac{1}{d+1}})$ & $\O(n^{1-\frac{1}{d+1}})$ \\
    \hline
    \multirow{3}{*}{$\epsilon$-} & \multirow{3}{*}{$\ell_p$}
    & $r$ is fixed  & $O(n)$ & $\O(\epsilon^{-d})$ & $\O(\epsilon^{-d})$ \\
    \cline{3-6}
    \multirow{3}{*}{Approximate} &  & $r$ is variable & \multirow{2}{*}{$O(\eps^{-d}n)$} & \multirow{2}{*}{$\O(\eps^{-d})$} & \multirow{2}{*}{$O(1)$} \\ 
     &  & 
    spread is $\poly(n)$ & & & \\
    \cline{2-6}
    & $\ell_1, \ell_2,$ & $r$ is fixed & \multirow{2}{*}{$\O(dn+n^{1+\rho})$} & \multirow{2}{*}{$\O(dn^{2\rho})$} & \multirow{2}{*}{$\O(dn^{2\rho})$} \\
    & hamming & high dimension &  & & \\
    \hline
    \end{tabular}
    
    \caption{Summary of Results: $n$ is the input size; $r$ is the distance threshold; $d$ is the dimension of input points; $\smash{\rho\leq \frac{1}{(1+\eps)^2}+o(1)}$ is the quality of LSH family for the $\ell_2$ metric. For $\ell_1$, Hamming $\rho\leq \frac{1}{1+\eps}$. $\O$ notation hides a $\log^{O(1)} n$-factor; for the results where $d$ is constant the $O(1)$ exponent is at most linear on $d$, while for the high dimensional case the exponent is at most $3$.}
    \label{tab:summary}
\end{table}
We present several dynamic data structures for enumerating similarity joins under different metrics. Table~\ref{tab:summary} summarizes our main results. It turns out that dynamic similarity join is hard for some metrics, e.g., $\ell_2$ metric. Therefore we also consider {\em approximate similarity join} where the distance threshold $r$ is a soft constraint. Formally, given parameter $r,\eps >0$, the \emph{$\eps$-approximate similarity join} relaxes the distance threshold for some parameter $\eps > 0$: (1) all pairs of $(a,b) \in A \times B$ with $\dist(a,b) \le r$ should be returned; (2) no pair of $(a,b)\in A \times B$ with $\dist(a,b)>(1+\eps)r$ is returned; (3) some pairs of $(a,b) \in A \times B$ with $r < \dist(a,b) \le (1+\eps)r$ may be returned.
We classify our results in four broad categories:

\paragraph{Exact similarity join.} 
Here we assume that $d$ is constant and the distance threshold is fixed. 
Our first result (Section~\ref{sec:linfty}) is an $\O(1)$-size data structure for similarity join under the $\ell_1/\ell_\infty$ metrics that 
can be updated in $\O(1)$ time whenever $A$ or $B$ is updated, and ensures $\O(1)$ time delay during enumeration. 
Based on range trees~\cite{bentley1978decomposable, de1997computational}, the data structure stores the similarity join pairs \emph{implicitly} so that 
they can be enumerated without probing every input point. 
We extend these ideas to construct a data structure for similarity join under the $\ell_2$ metric (in Section~\ref{sec:l2}) 
with $\O(n^{1-1/d})$ amortized update time while supporting $\O(n^{1-1/d})$-delay enumeration. Lower bounds on ball range searching~\cite{afshani2012improved, chazelle1996simplex} rule out the possibility of a linear-size data structure with $\O(1)$ delay. 

\paragraph{Approximate similarity join in low dimensions.}  
Due to the negative result for $\ell_2$ metric, we shift our attention to $\eps$-approximate similarity join. We now allow the distance threshold to be part of the query but the value of $\eps$, the error parameter, is fixed. We present a simple linear-size data structure based on quad trees and the notion of well-separated pair decomposition, with $O(\epsilon^{-d})$ update time and $O(1)$ delay. 
If we fix the distance threshold, then the data structure can be further simplified and somewhat improved by replacing the quad tree with a simple uniform grid. 

\paragraph{Approximate similarity join in high dimensions.} So far we assumed $d$ to be constant and the big $O$ notation in some of the previous bounds hides a constant that is exponential in $d$.  
	Our final result is an  LSH-based~\cite{gionis1999similarity} data structure for similarity joins in high dimensions. 
	Two technical issues arise when enumerating join results from LSH: one is to ensure bounded delay because we do not want to enumerate false positive results identified by the hash functions,
	and the other is to remove duplicated results as one join result could be identified by multiple hash functions.
    For the $\ell_2$ metric (the results can also be extended to $\ell_1$, Hamming metrics) we propose a data structure of $\O(nd +n^{1+\rho})$ size and $\O(dn^{2\rho})$ amortized update time that supports $(1+2\eps)$-approximate enumeration with $\O(dn^{2\rho})$ delay with high probability, where $\rho \le \frac{1}{(1+\eps)^2}+o(1)$ is the quality of the LSH family.
   %
   Alternatively, we present a data structure with $\O(dn^{\rho})$ amortized update time and $\O(dn^{3\rho})$ delay.
    Our data structure can be extended to the case when the distance threshold $r$ is variable. If we allow worse approximation error we can improve the results for the Hamming distance. Finally, we show a lower bound by relating similarity join to the {\em approximate nearest neighbor} query.
	%

    We also consider similarity join beyond binary joins. 
	
	\paragraph{Triangle similarity join in low dimensions.}  
	Given three sets of points $A, B, S$ in $\Re^d$, a metric $\dist(\cdot)$, and a distance threshold $r > 0$, the {\em triangle similarity join} asks to report 
	the set of all triples of $(a,b,s) \in A \times B \times S$ with $\dist(a,b) \le r, \dist(a,s) \le r, \dist(b,s) \le r$. The {\em $\eps$-approximate triangle similarity join} can be defined similarly by taking the distance threshold $r$ as a soft constraint. We extend our data structures 
	to approximate {\em triangle similarity join} by paying a factor of $\log^{O(1)} n$ in the performance. 

    \paragraph{High-level framework.}
    All our data structures rely on the following common framework. We model the (binary) similarity join as a bipartite graph $\graph = (A\cup B, E)$, where an edge $(a,b) \in E$ if and only if $\dist(a,b) \le r$. A naive solution by maintaining all edges of $\graph$ explicitly leads to a data structure of $\Theta(n^2)$ size that can be updated in $\Theta(n)$ time while supporting $O(1)$-delay enumeration. To obtain a data structure with poly-logarithmic update time and delay enumeration, we find a compact representation of $\graph$ with a set $\setCliques=\{(A_1, B_1), (A_2, B_2),\ldots, (A_u,B_u)\}$ of edge-disjoint bi-cliques such that (i) $A_i\subseteq A$, $B_i\subseteq B$ for any $i$, (ii) $E = \bigcup_{i=1}^u A_i \times B_i$, and (iii) $(A_i \times B_i) \cap (A_j \times B_j) = \emptyset$ for any $i \neq j$. We represent $\setCliques$ using a tripartite graph $\mathcal{G}=(A\cup B\cup C, E_1\cup E_2)$ where $C=\{c_1, \ldots, c_u\}$ has a node for each bi-clique in $\setCliques$ and for every $i\leq u$, we have the edges $(a_j,c_i)\in E_1$ for all $a_j\in A_i$ and $(b_k,c_i)\in E_1$ for all $b_k\in B_i$. We \emph{cannot} afford to maintain $E_1$ and $E_2$ explicitly. Instead, we store some auxiliary information for each $c_i$ and use geometric data structures to recover the edges incident to a vertex $c_i\in C$. We also use data structures to maintain the set $C$ and the auxiliary information dynamically as $A$ and $B$ are being updated. We will not refer to this framework explicitly but it provides the intuition behind all our data structures. Section~\ref{sec:exact} describes the data structures to support this framework for exact similarity join, and Section~\ref{sec:approximate} presents simpler, faster data structures for approximate similarity join. Both Sections~\ref{sec:exact} and~\ref{sec:approximate} assume $d$ to be constant. Section~\ref{sec:highd} describes the data structure for approximate similarity join when $d$ is not constant.

%% file: Linfty.tex
\section{Exact Similarity Join}
\label{sec:exact}
In this section, we describe the data structure for exact similarity joins under the $\ell_\infty, \ell_1, \ell_2$ metrics, assuming $d$ is constant. We first describe the data structure for the $\ell_{\infty}$ metric. We show that similarity join under the $\ell_1$ metric in $\Re^d$ can be reduced to that under the $\ell_{\infty}$ metric in $\Re^{d+1}$. Finally, we describe the data structure for the $\ell_2$ metric. Throughout this section, the threshold $r$ is fixed, which is assumed to be $1$ without loss of generality.

\subsection{Similarity join under $\ell_\infty$ metric}
\label{sec:linfty}

Let $A$ and $B$ be two point sets in $\Re^d$ with $|A|+|B|=n$.  
For a point $p\in\Re^d$, let $\mathcal{B}(p)=\{x\in \Re^d\mid \norm{p-x}_{\infty}\leq 1\}$ be the hypercube of side length $2$. We wish to enumerate pairs $(a,b)\in A\times B$ such that $a\in \mathcal{B}(b)$.


\paragraph{Data structure.}
We build a $d$-dimensional dynamic range tree $\T$ on the points in $A$.
For $d=1$, the range tree on $A$ is a balanced binary search tree $\T$ of $O(\log n)$ height. The points of $A$ are stored at the leaves of $\T$ in increasing order, while each internal node $v$ stores the smallest and the largest values, $\alpha_v^-$ and $\alpha_v^+$, respectively, contained in its subtree.
The node $v$ is associated with an interval $I_v=[\alpha_v^-, \alpha_v^+]$ and the subset $A_v=I_v\cap A$.
For $d>1$, $\T$ is constructed recursively: 
We build a $1$D range tree $\T_d$ on the $x_d$-coordinates of points in $A$. Next, for each node $v\in \T_d$, we recursively construct a $(d-1)$-dimensional range tree $\T_v$ on $A_v^*$, which is defined as the projection of $A_v$ onto the hyperplane $x_d=0$, and attach $\T_v$ to $v$ as its secondary tree. The size of $\T$ in $\Re^d$ is $O(n\log^{d-1} n)$ and it can be constructed in $O(n\log^d n)$ time.
See~\cite{de1997computational} for details.

For a node $v$ at a level-$i$ tree, let $p(v)$ denote its parents in that tree. If $v$ is the root of that tree, $p(v)$ is undefined.
For each node $u$ of the $d$-th level of $\T$, we associate a $d$-tuple $\pi(u)=\langle u_1, u_2, \ldots, u_d=u\rangle$, where $u_i$ is the node at the $i$-th level tree of $\T$ to which the level-$(i+1)$ tree containing $u_{i+1}$ is connected.
We associate the rectangle $\square_u=\prod_{j=1}^d I_{u_j}$ with the node $u$.
For a rectangle $\rho=\prod_{i=1}^d \delta_i$
, a $d$-level node is called a \emph{canonical node} if for every $i\in [1,d]$, $I_{u_i}\subseteq \delta_i$ and $I_{p(u_i)}\not\subseteq \delta_i$.
For any rectangle $\rho$, there are $O(\log^d n)$ canonical nodes in $\T$, denoted by $\mathcal{N}(\rho)$, and they can be computed in $O(\log^d n)$ time~\cite{de1997computational}.
$\T$ can be maintained dynamically, as points are inserted into $A$ or deleted from $A$ using the standard partial-reconstruction method, which periodically reconstructs various bottom subtrees. The amortized time is $O(\log^d n)$; see~\cite{overmars1987design} for details.

We query $\T$ with $\mathcal{B}(b)$ for all $b\in B$ and compute $\mathcal{N}(b):=\mathcal{N}(\mathcal{B}(b))$ the sets of its canonical nodes. For each level-$d$ tree node $u$ of $\T$, let $B_u=\{b\in B\mid u\in \mathcal{N}(b)\}$. We have $\sum_{u}|B_u|=O(n\log^d n)$. By construction, for all pairs $(a,b)\in A_u\times B_u$, $\norm{a-b}_{\infty}\leq 1$, so $(A_u, B_u)$ is a bi-clique of 
join results. We call $u$ \emph{active} if both $A_u, B_u\neq \emptyset$.
A naive approach for reporting join results 
is to maintain $A_u, B_u$ for every $d$-level node $u$ of $\T$ as well as the set $\C$ of all active nodes. Whenever an enumerate query is issued, we traverse $\C$ and return $A_u\times B_u$ for all $u\in \C$ (referring to the tripartite-graph framework mentioned in Introduction, $C$ is the set of all level-$d$ nodes of $\T$).
The difficulty with this approach is that when $A$ changes and $\T$ is updated, some $d$-level nodes change and we have to construct $B_u$ for each new level-$d$ node $u\in \T$. It is too expensive to scan the entire $B$ at each update. Furthermore, although the average size of $B_u$ is small, it can be very large for a particular $u$ and this node may appear and disappear several times. So we need  a different approach. The following lemma is the key observation.

\begin{lemma}
\label{lem:A}
Let $u$ be a level-$d$ node, and let $\pi(u)=\langle u_1, \ldots, u_d=u\rangle$. Then there is a $d$-dimensional rectangle $\mathcal{R}(u)=\prod_{i=1}^d\delta_i$, where the endpoints of $\delta_i$, for $i\in [1,d]$, are defined by the endpoints of $I_{u_i}$ and $I_{p(u_i)}$, such that for any $x\in \Re^d$, $u\in \mathcal{N}(x)$ if and only if $x\in \mathcal{R}(u)$. Given $u_i$'s and $p(u_i)$'s, $\mathcal{R}(u)$ can be constructed in $O(1)$ time.
\end{lemma}
\begin{figure}[t]
    \centering
    \includegraphics[scale=0.4]{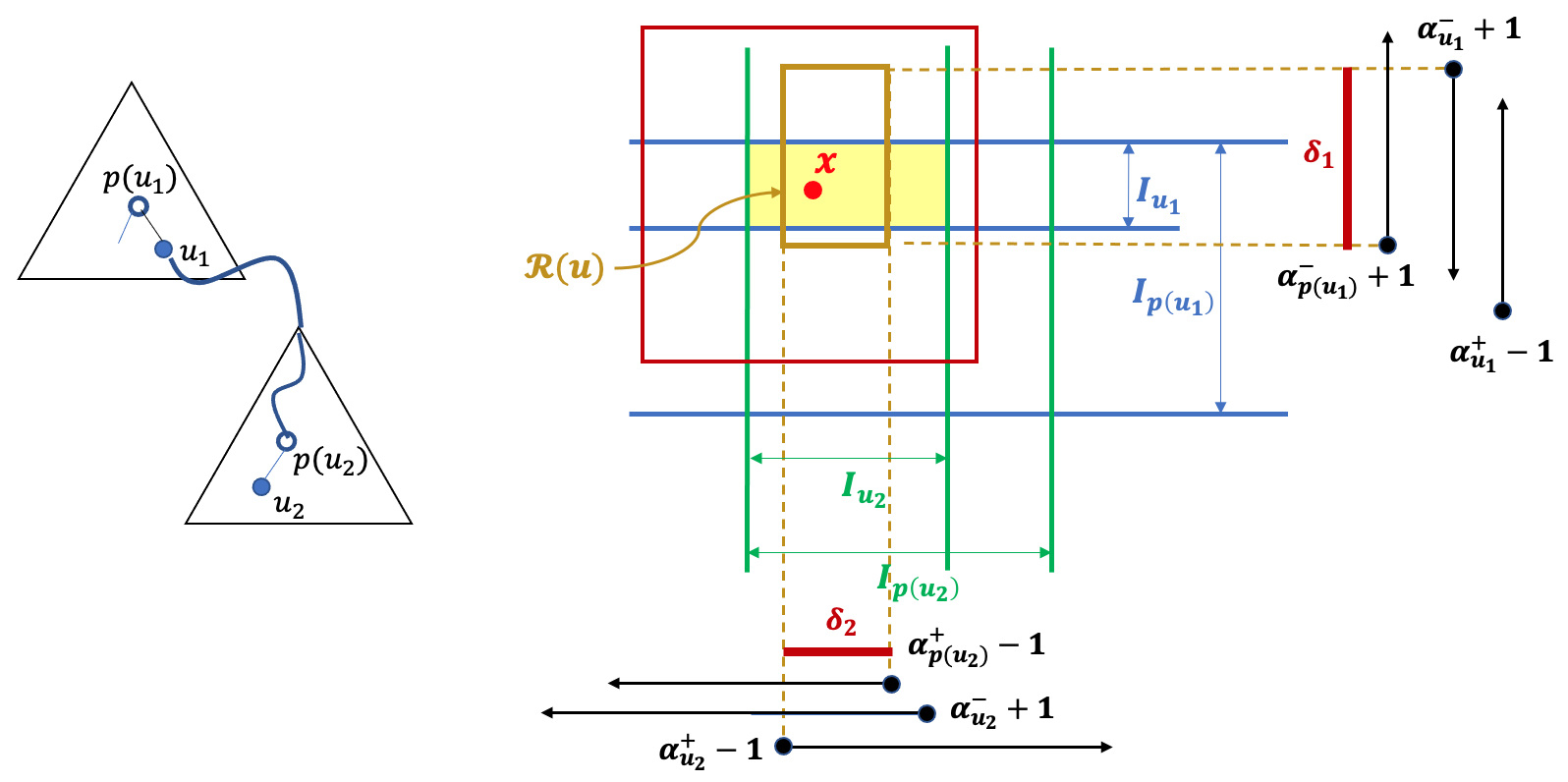}
    \caption{Left: Two levels of the range tree. Right: Definition of $\mathcal{R}(u)$.}
    \label{fig:example}
\end{figure}
\begin{proof}
Notice that $\mathcal{B}(x)$ is the hypercube of side length $2$ and center $x$.
Let $I_{u_i}=[\alpha_{u_i}^-, \alpha_{u_i}^+]$ for any $u_i$ and $i\in[1,d]$.
Recall that $u\in \mathcal{N}(x)$ if and only if for each $i\in [1,d]$,  \[I_{u_i}\subseteq [x_i-1, x_i+1] \textrm{ and }  I_{p(u_i)}\not\subseteq [x_i-1, x_i+1], \ \ \ (*)\] 
Fix a value of $i$.
From the construction of a range tree either $\alpha_{u_i}^-=\alpha_{p(u_i)}^-$ or $\alpha_{u_i}^+=\alpha_{p(u_i)}^+$.
Without loss of generality, assume $\alpha_{u_i}^-=\alpha_{p(u_i)}^-$; the other case is symmetric.
Then ($*$) can be written as: 
$x_i\leq \alpha_{u_i}^-+1$ and $\alpha_{u_i}^+-1\leq x_i<\alpha_{p(u_i)}^+-1$.
Therefore $x_i$ has to satisfy three $1$D linear constraints. The feasible region of these constraints is an interval $\delta_i$ and $x_i\in \delta_i$ (see also Figure~\ref{fig:example}).
Hence, $u$ is a canonical node of $\mathcal{B}(x)$ if and only if for all $i\in [1,d]$, $x_i\in \delta_i$. In other words, $x=(x_1,\ldots, x_d)\in \prod_{i=1}^d\delta_i := \mathcal{R}(u)$. The endpoints of $\delta_i$ are the endpoints of $I_{u_i}$ or $I_{p(u_i)}$.
%
%
In order to construct $\mathcal{R}(u)$, we only need the intervals $I_{u_i}$ and $I_{p(u_i)}$ for each $i\in [1,d]$, so it can be constructed in $O(d)=O(1)$ time.
%
%
\end{proof}

In view of Lemma~\ref{lem:A}, we proceed as follows. We build a dynamic range tree $\mathcal{Z}$ on $B$. Furthermore, we augment the range tree $\T$ on $A$ as follows. For each level-$d$ node $u\in \T$, we compute and store $\mathcal{R}(u)$ and $\beta_u=|B_u|$. By construction, $|A_u|\geq 1$ for all $u$. We also store  a pointer at $u$ to the leftmost leaf of the subtree of $\T$ rooted at $u$, and we thread all the leaves of a $d$-level tree so that for a node $u$, $A_u$ can be reported in $O(|A_u|)$ time. Updating these pointers as $\T$ is updated is straightforward. Whenever a new node $u$ of $\T$ is constructed, we query $\mathcal{Z}$ with $\mathcal{R}(u)$ to compute $\beta_u$. Finally, we store $\C$, the set of all active nodes of $\T$, in a red-black tree so that a node can     be inserted or deleted in $O(\log n)$ time. The total size of the data structure is  $O(n\log^{d-1} n)$, and it can be constructed in $O(n\log^d n)$ time.

\paragraph{Update and Enumerate.}
Updating $A$ is straightforward. We update $\T$, query $\mathcal{Z}$ with $\mathcal{R}(u)$, for all newly created $d$-level nodes $u$ in $\T$ to compute $\beta_u$, and update $\C$ to delete all active nodes that are no longer in $\T$ and to insert new active nodes. Since the amortized time to update $\T$ as a point is inserted or deleted is $O(\log^d n)$, the amortized update time of a point in $A$ is $O(\log^{2d} n)$ --- we spend $O(\log^d n)$ time to compute $\beta_u$ for each of $O(\log^d n)$ newly created nodes. If a point $b$ is inserted  (resp. deleted) in $B$, we update $\mathcal{Z}$ and query $\T$ with $\mathcal{B}(b)$. For all canonical nodes $u$ in $\mathcal{N}(b)$, we increment (resp. decrement) $b_u$. If $u$ becomes active (resp. inactive), we insert (resp. delete) $u$ in $\C$ in $O(\log n)$ time. The amortized update time for $b$ is $O(\log^{d+1} n)$.

Finally, to enumerate the pairs in join results, 
we traverse the active nodes $\C$ and for each $u\in \C$, we first query $\mathcal{Z}$ with $\mathcal{R}(u)$ to recover $B_u$.
Recall that $B_u$ is reported as a set of $O(\log^d n)$ canonical nodes of $\mathcal{Z}$ whose leaves contain the points of $B_u$.
We simultaneously traverse the leaves of the subtree of $\T$ rooted at $u$ to compute $A_u$ and report $A_u\times B_u$. The traversals can be performed in $O(\log^{d} n)$ maximum delay. Putting everything together, we obtain:

\begin{theorem}
\label{the:rectangle-point}
Let $A, B$ be two sets of points in $\Re^d$, where $d\geq 1$ is a constant, with $|A|+|B|=n$. A data structure of $\O(n)$ size can be built in $\O(n)$ time and updated in $\O(1)$ amortized time, while supporting $\O(1)$-delay enumeration of similarity join under $\ell_{\infty}$ metric.
\end{theorem}

\subsection{Similarity join under $\ell_1$ metric}
For $d\leq 2$ it is straightforward to reduce similarity join under $\ell_1$ metric to $\ell_\infty$ metric. For $d=1$, $\ell_1$ metric is obviously equivalent to the $\ell_\infty$ metric. For $d=2$, notice that the $\ell_1$ ball is a diamond, while the $\ell_\infty$ ball is a square. Hence, given an instance of the similarity join under the $\ell_1$ metric we can rotate $A\cup B$ by $45$ degrees to create an equivalent instance of the similarity join problem under the $\ell_\infty$ metric.

Next, we focus on $d\geq 3$. The data structure we proposed in Section~\ref{sec:linfty} for the $\ell_\infty$ norm can be straightforwardly extended to the \emph{rectangle-containment} problem in which for each $b\in B$, $\mathcal{B}(b)$ is an arbitrary axis-aligned hyper-rectangle with center $b$, and the goal is to report all $(a,b)\in A\times B$ such that $a\in \mathcal{B}(b)$. Lemma~\ref{lem:A} can be extended so that $\mathcal{R}(u)$ is a $2d$-dimensional rectangle. Overall, Theorem~\ref{the:rectangle-point} remains the same assuming $\mathcal{B}(b)$ are hyper-rectangles (and not hypercubes).

Given an instance of similarity join under $\ell_1$ metric in $\Re^d$, we next show how to reduce it to $2^d$ $(d+1)$-dimensional rectangle-containment problems.
As above, assume $r= 1$, so our goal is to report all pairs $a=(a_1,\ldots, a_d)\in A$, $b=(b_1,\ldots, b_d)\in B$ such that $\sum_{i=1}^d|a_i-b_i|\leq 1$.


Let $E=\{-1,+1\}^d$ be the set of all $2^d$ vectors in $\Re^d$ with coordinates either $1$ or $-1$.
For each vector $e\in E$, we construct an instance of the rectangle-containment problem.
For each $e=(e_1,\ldots, e_d)\in E$, we map each point $a=(a_1, \ldots, a_d)\in A$ to a point $\bar{a}_e=(a_1, \ldots, a_d, \sum_{i=1}^d e_{i}a_i) \in \Re^{d+1}$.
Let $\bar{A}_e=\{\bar{a}_e\mid a\in A\}$.
For each point $b=(b_1, \ldots, b_d)\in B$, we construct the axis-align rectangle $\bar{b}_e=\prod_{i=1}^{d+1} b_e^{(i)}$
in $\Re^{d+1}$, where $b_e^{(i)}$ is the interval $[b_i, \infty)$ if $e_{i}=1$
and $(-\infty, b_i]$ if $e_{i}=-1$
for each $i=1,\ldots, d$, and $b_e^{(d+1)}=(-\infty, 1+\sum_{i=1}^d e_{i}b_{i}]$. Let $\bar{B}_e=\{\bar{b}_e\mid b\in B\}$. See Figure~\ref{fig:l1transformation}.

\begin{figure*}
    \centering
    \includegraphics[scale=0.7]{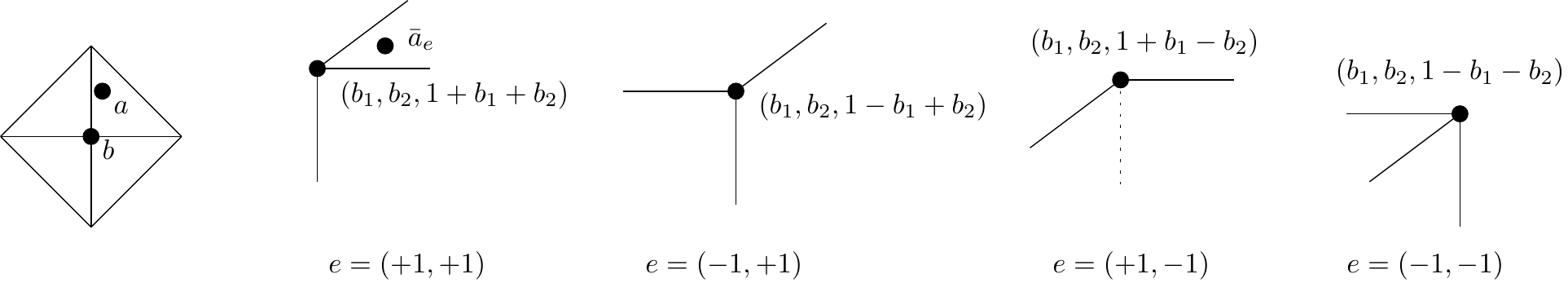}
    \caption{An illustration of mapping each $b$ to rectangles.}
    \label{fig:l1transformation}
\end{figure*}

For each $e \in E$, we construct the dynamic data structure for $\bar{A}_e, \bar{B}_e$.
Whenever $A$ or $B$ is updated, we update all $2^d$ rectangle-containment data structures.
A similarity join enumeration query on $A, B$ is answered by enumerating containment pairs $(\bar{A}_e, \bar{B}_e)$ for $e\in E$. If a pair $(\bar{a}_e, \bar{b}_e)$ is reported, we report $(a,b)$. The update time and delay are $\O(1)$. The correctness of the algorithm follows from the following lemma. Let $\sign(x)=+1$ if $x\geq 0$ and $-1$ otherwise.

\begin{lemma}
\label{lem:ell1}
Let $a=(a_1,\ldots, a_d)\in A$, $b=(b_1,\ldots, b_d)\in B$ be an arbitrary pair of points. Let $e^*=(e^*_1,\ldots, e^*_d)$ where $e^*_i=\sign(a_i-b_i)$ for $1\leq i\leq d$. Then $\bar{a}_e\notin \bar{b}_e$ for all $e\notin E\setminus\{e^*\}$. Furthermore, $\bar{a}_{e^*}\in \bar{b}_{e^*}$ if and only if $\lVert a-b \rVert_1\leq 1$.
\end{lemma}
\begin{proof}
First, we note that for any $e \in E \setminus \{e^*\}$, there must exist some $i$ such that $e_i \neq e^*_i$. Without loss of generality, assume $e_j = 1$ when $a_j < b_j$. By the definition of $\bar{a}_e, \bar{b}_e$, $a_j\notin [b_j, \infty)$, thus $\bar{a}_e \notin \bar{b}_e$. Next, we show that $\bar{a}_{e^*}\in \bar{b}_{e^*}$ if and only if $\lVert a-b \rVert_1\leq 1$.
On one hand, we assume $\bar{a}_{e^*}\in \bar{b}_{e^*}$. By definition, $\sum_{i=1}^d e^*_i a_i$ lies in the interval associated with $b^{d+1}_{e^*}$, i.e., $\sum_{i=1}^d e^*_i a_i \le 1 + \sum_{i=1}^d e^*_i b_i$, or $\sum_{i=1}^d e^*_i (a_i -b_i) \le 1$. Implied by the fact that $\lVert a-b\rVert_1 = \sum_{i=1}^d e^*_{i}(a_i-b_i)$, we have $\lVert a-b\rVert_1 \le 1$. 
On the other hand, assume $\lVert a-b\rVert_1 \le 1$. Similarly, we have $ \lVert a-b\rVert_1 = \sum_{i=1}^d e^*_{i}(a_i-b_i) \le 1 \Leftrightarrow \sum_{i=1}^d e^*_{i}a_i \le 1 + \sum_{i=1}^d e^*_{i} b_i$, or $\sum_{i=1}^d e^*_{i}a_i \in (-\infty, 1 + \sum_{i=1}^d e^*_{i} b_i]$. Moreover, for any $i \in \{1,\ldots, d\}$, we have: (1) if $e^*_i = 1$, $a_i \ge b_i$, i.e., $a_i \in [b_i, \infty)$; (2) if $e^*_i = -1$, $a_i \le b_i$, i.e., $a_i \in (-\infty, b_i]$. Hence, $\bar{a}_{e^*} \in \bar{b}_{e^*}$. 
%
\end{proof}

\begin{figure}[h]
    \centering
    \includegraphics[scale=0.4]{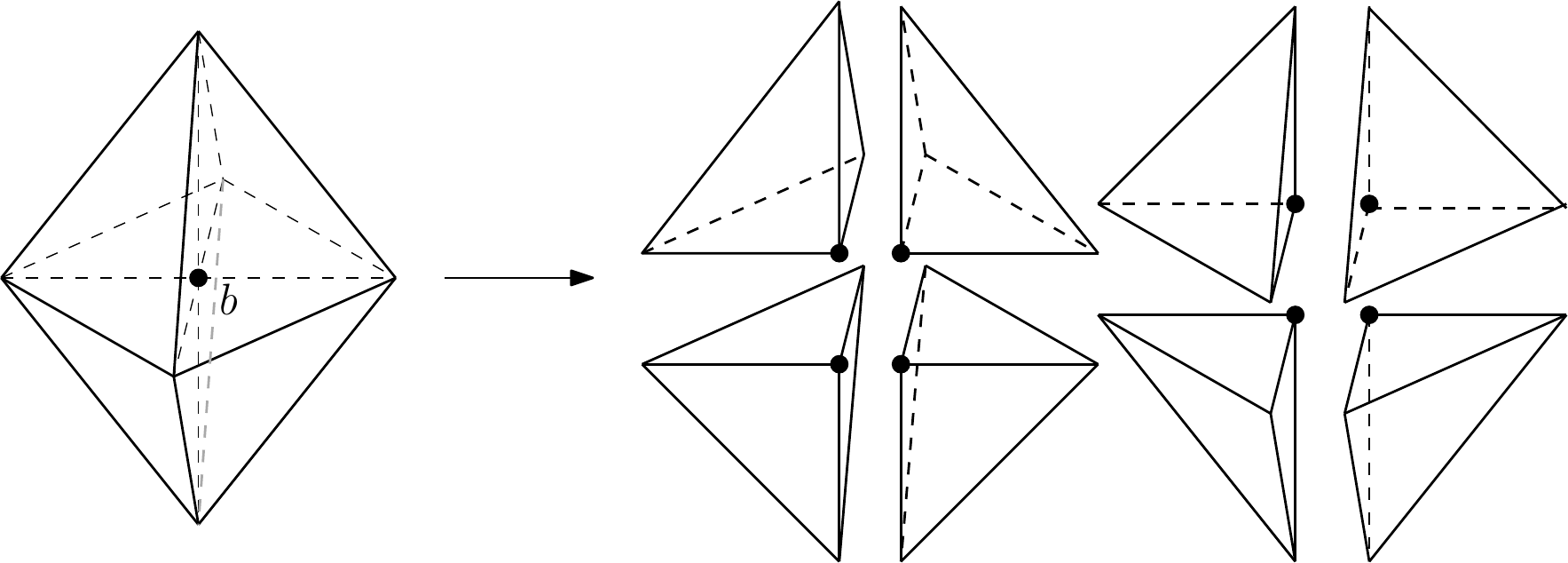}
    \caption{An illustration of $\ell_1$ ball in $\Re^3$. It is decomposed to $2^d=8$ types of simplices.}
    \label{fig:l1ball}
\end{figure}

\paragraph{Remark.}
Roughly speaking, we partition the $\ell_1$-ball centered at $0$ into $2^d$ simplices $\Delta_1,\ldots, \Delta_{2^d}$ (see Figure~\ref{fig:l1ball}) and build a separate data structure for each simplex $\Delta_i$. Namely, let $\mathcal{B}_i=\{b+\Delta_i\mid b\in B\}$ and we report all pairs $(a,b)\in A\times B$ such that $a\in b+\Delta_i$. If $\lVert a-b \rVert_1\leq 1$ then $a$ lies in exactly one simplex $b\in \Delta_i$. We map each simplex to a rectangle in $\Re^{d+1}$ and use the previous data structure.

\medskip Using Theorem~\ref{the:rectangle-point}, we obtain:


\begin{theorem}
\label{the:pyramid-point}
Let $A, B$ be two sets of points in $\Re^d$, where $d\geq 1$ is a constant, with $|A|+|B|=n$. A data structure of $\O(n)$ size can be built in $\O(n)$ time and updated in $\O(1)$ amortized time, while supporting $\O(1)$-delay enumeration of similarity join under $\ell_{1}$ metric.
\end{theorem}

%% file: L2.tex
\subsection{Similarity join under $\ell_2$ metric}
\label{sec:l2}

In this section, we consider the similarity join between two point sets $A$ and $B$ in $\Re^d$ under the $\ell_2$ metric.


\paragraph{Reduction to halfspace containment.}
We use the {\em lifting transformation}~\cite{de1997computational} to convert 
an instance of the similarity join problem under $\ell_2$ metric to the halfspace-containment problem in $\Re^{d+1}$.
For any two points $a=(a_1,\ldots, a_d)\in A$ and $b=(b_1, \ldots, b_d)\in B$, $\lVert a-b \rVert_2\leq 1$ if and only if 
$(a_1-b_1)^2+\ldots +(a_d-b_d)^2\leq 1$, or $a$ lies in the unit sphere centered at $b$.
The above condition can be rewritten as
\[ a_1^2+b_1^2 +\cdots + a_d^2 + b_d^2 - 2a_1 b_1 - \cdots - 2a_d b_d - 1 \ge 0.\]
We map the point $a$ to a point $a'=(a_1, \dots, a_d, a_1^2 + \cdots + a_d^2)$ in $\Re^{d+1}$ and the point $b$ to a halfspace $b'$ in $\Re^{d+1}$ defined as 
\[b': -2 b_1 z_1 - \cdots - 2b_d z_d +z_{d+1}+b_1^2 +\cdots + b_d^2 - 1\ge 0.\]
Note that $\lVert a-b \rVert_2\leq 1$ if and only if $a'\in b'$.
Set
$A'=\{a'\mid a\in A\}$ and $B'=\{b'\mid b\in B\}$.
Thus, in the following, we study the halfspace-containment problem, where given a set of points $A'$ and a set of halfspaces $B'$ we construct a dynamic data structure that reports all pairs $(a\in A', b\in B')$, such that $a$ belongs in the halfspace $b$, with delay guarantee.


\paragraph{Partition tree.}
A partition tree on a set $P$ of points in $\Re^d$~\cite{chan2012optimal,matouvsek1992efficient, willard1982polygon} is a tree data structure formed by recursively partitioning a set into subsets.
Each point is stored in exactly one leaf and each leaf usually contains a constant number of points.
Each node $u$ of the tree is associated with
a simplex $\Delta_u$ and the subset $P_u=P\cap \Delta_u$; the subtree rooted at $u$ is a partition tree of $P_u$. We assume that the simplices associated with the children of a node $u$ are pairwise disjoint and lie inside $\Delta_u$, as in~\cite{chan2012optimal}.
In general, the degree of a node is allowed to be non-constant.
Given a query simplex $\Delta$, a partition tree finds a set of $O(n^{1-1/d})$ \emph{canonical} nodes whose cells contain the points of $P\cap \Delta$.
Roughly speaking, a node $u$ is a canonical node for $\Delta$ if $\Delta_u\subset \Delta$ and $\Delta_{p(u)}\not\subseteq \Delta$.
A simplex counting (resp. reporting) query can be answered in $O(n^{1-1/d})$ (resp. $O(n^{1-1/d}+k)$) time using a partition tree.
Chan~\cite{chan2012optimal} proposed a randomized algorithm for constructing a linear size partition tree with constant degree, 
that runs in $O(n\log n)$ time and it has $O(n^{1-1/d})$ query time with high probability.

\paragraph{Data structure.}
For simplicity, with slight abuse of notation, let $A$ be a set of points in $\Re^d$ and $B$ a set of halfspaces in $\Re^d$ each lying below the hyperplane bounding it, and our goal is to build a dynamic data structure for halfspace-containment join on $A,B$. The overall structure of the data structure is the same as for rectangle containment described in Section~\ref{sec:linfty}, so we simply highlight the difference.

Instead of constructing a range tree, we construct a dynamic partition tree $\T_A$ for $A$ so that the points of $A$ lying in a halfspace can be represented as the union of $O(n^{1-1/d})$ canonical subsets.
For a halfplane bounding a halfspace $b\in B$, let $\bar{b}$ denote its dual point in $\Re^d$ (see \cite{de1997computational} for the definition of duality transform). Note that a point $a$ lies in $b$ if and only if the dual point $\bar{b}$ lies in the halfspace lying below the hyperplane dual to $a$. Set $\bar{B}=\{\bar{b}\mid b\in B\}$. We construct a multi-level dynamic partition tree on $\bar{B}$, so that for a pair of simplices $\Delta_1$ and $\Delta_2$, it returns the number of halfspaces of $B$ that satisfy the following two conditions: (i) $\Delta_1\subseteq b$ and (ii) $\Delta_2\cap \partial b\neq \emptyset$, where $\partial b$ is the hyperplane boundary defined by the halfspace $b$. This data structure uses $O(n)$ space, can be constructed in $\O(n)$ time, and answers a query in $\O(n^{1-1/d})$ time.

For each node $u \in \T_A$, we issue a counting query to $\T_B$ and get the number of halfspaces in $B$ that have $u$ as a canonical node. Hence, $\T_A$ can be built in $\O(n^{2-1/d})$ time. For a node $u$, $\mu_A(u)$ can be computed in $O(1)$ time by storing $A_u$ at each node $u\in \T_A$. Recall that $\mu_B(u)$ is the number of halfspaces $b$ of $B$ for which $u$ is a canonical node, i.e., $\Delta_u \subseteq b$ and $\Delta_{p(u)}\cap \partial b\neq \emptyset$, where $p(u)$ is the parent of $u$. Using $\T_B$, $\mu_B(u)$ can be computed in $\O(n^{1-1/d})$ time.

\paragraph{Update and enumeration.} The update procedure is the same that in Section~\ref{sec:linfty}, however the query time now on $\T_A$ or $\T_B$ is $\O(n^{1-\frac{1}{d}})$ so the amortized update time is $\O(n^{1-\frac{1}{d}})$.
The enumeration query is also the same as in Section~\ref{sec:linfty} but a reporting query in $\T_B$ takes $\O(n^{1-\frac{1}{d}}+k)$ time (and it has delay at most $\O(n^{1-\frac{1}{d}})$), so the overall delay is $\O(n^{1-\frac{1}{d}})$.

\begin{theorem}
\label{thm:points-halfspaces}
Let $A$ be a set of points and $B$ be a set of half-spaces in $\mathbb{R}^d$ with $|A| + |B| = n$. A data structure of $\O(n)$ size can be built in $
\O(n^{2-\frac{1}{d}})$ time and updated in $\O(n^{1-\frac{1}{d}})$ amortized time while supporting $\O(n^{1-\frac{1}{d}})$-delay enumeration of halfspace-containment query.
\end{theorem}

Using Theorem~\ref{thm:points-halfspaces} and the lifting transformation described at the beginning of this section we conclude with Corollary~\ref{cor:ell2}.

\begin{corollary}
\label{cor:ell2}
Let $A, B$ be two sets of points in $\Re^d$, where $d\geq 1$ is a constant, with $|A|+|B|=n$. A data structure of $\O(n)$ size can be constructed in $
\O(n^{2-\frac{1}{d+1}})$ time and updated in $\O(n^{1-\frac{1}{d+1}})$ amortized time, while supporting $\O(n^{1-\frac{1}{d+1}})$-delay enumeration of similarity join under the $\ell_{2}$ metric.
\end{corollary}

\paragraph{Lower bound.} 
We show a lower bound for the similarity join in the pointer-machine model under the $\ell_2$ metric based on the hardness of unit sphere reporting problem.
Let $P$ be a set of $n$ points in $\Re^d$ for $d>3$. The unit-sphere reporting problem asks for a data structure on the points in $P$, such that given any unit-sphere $b$ report all points of $P\cap b$. If the space is $\O(n)$, 
it is not possible to get a data structure for answering unit-sphere reporting queries in $\O(k+1)$ time in the pointer-machine model, where $k$ is the output size for $d\geq 4$
\cite{afshani2012improved}. 

For any instance of sphere reporting problem, we construct an instance of similarity join over two sets, with $A = \emptyset$, $B = P$, and $r=1$. Given a query unit-sphere of center $q$, we insert point $q$ in $A$, issue an enumeration query, and then remove $q$ from $A$. All results enumerated (if any) are the results of the sphere reporting problem. If there exists a data structure for enumerating similarity join under $\ell_2$ metric using $\O(n)$ space, with $\O(1)$ update time and $\O(1)$ delay, we would break the barrier. 

\begin{theorem}
\label{theorem:lowerboundSphere}
Let $A, B$ be two sets of points in $\mathbb{R}^d$ for $d > 3$, with $|A| + |B| = n$.
If using $\O(n)$ space, there is no data structure under the pointer-machine model that can be updated in $\O(1)$ time, while supporting $\O(1)$-delay enumeration of similarity join under the $\ell_{2}$ metric.
\end{theorem}



\section{Approximate Enumeration}
\label{sec:approximate}

In this section we propose a dynamic data structure for answering approximate similarity-join queries under any $\ell_p$ metric. For simplicity, we use the $\ell_2$ norm to illustrate the main idea and assume $\dist(a,b)=||a-b||_2$. Recall that all pairs of $(a,b)\in A\times B$ with $\dist(a,b)\leq r$ must be reported, along with (potentially) some pairs of $(a', b')$ with $\dist(a',b')\leq (1+\eps)r$, but no pair $(a,b)$ with $\dist(a,b)>(1+\eps)r$ is reported.

We will start with the setting where the distance threshold $r$ is not fixed and specified as part of a query, and then move to a simpler scenario where $r$ is fixed.



\begin{figure}
\centering
  \minipage{0.45\textwidth}
  \centering
  \vspace{3.2em}
  \includegraphics[scale=0.6]{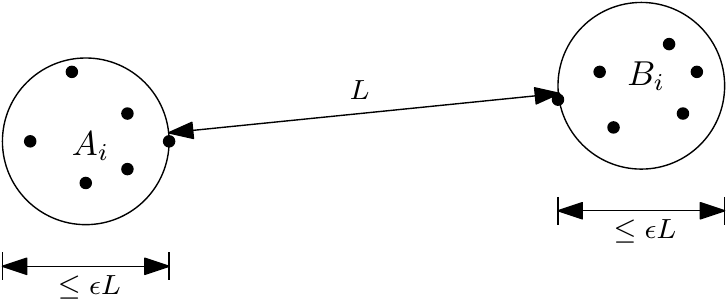}
  \vspace{0em}
    \caption{An example pair of $\eps$-WSPD.}
    \label{fig:wspd}
  \endminipage\hfill 
  \minipage{0.55\textwidth}
  \centering
  \includegraphics[scale=0.5]{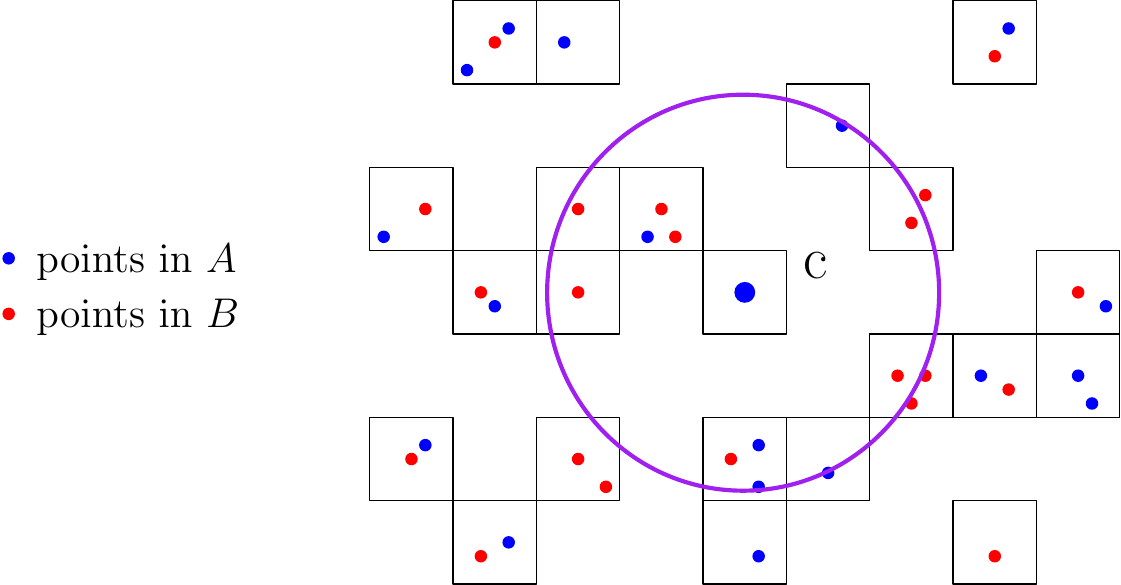}
  \caption{
  An example of active cell $c$ in the grid.
  }
  \label{fig:Grid}
\endminipage
\end{figure}

\subsection{Variable Similarity Threshold}
\label{sec:wspd}

We describe the data structure when $r$ is part of the query. In this subsection we assume that the spread of $A\cup B$ is polynomially bounded, i.e., $sp(A\cup B)=\frac{\max_{p,q \in A\cup B}\dist(p,q)}{\min_{p\neq q\in A\cup B}\dist(p,q)}=n^{O(1)}$. We use a quad tree and well-separated pair decomposition (WSPD) for our data structure. We describe them briefly here and refer the reader to~\cite{har2011geometric, samet1989spatial} for details.

\paragraph{Quad tree and WSPD.}
A $d$-dimensional quad tree over a point set $P$ is a tree data structure $\T$ in which each node $u$ is associated with a hypercube $\square_u$ in $\Re^d$, called a \emph{cell}, and each internal node has $2^d$ children.
The root is associated with a hypercube containing $P$. For a node $u$, let $P_u=P\cap \square_u$. A node $u$ is a leaf if $|P_u|\leq 1$.
The tree recursively subdivides the space into $2^d$ congruent hypercubes until a box contains at most one point from $P$.
If $sp(P)=n^{O(1)}$, the height of $\T$ is $O(\log n)$.

Given two point sets $A, B\subset \Re^d$, with $|A|+|B|=n$,
and a parameter $0< \eps < \frac{1}{2}$, a family of pairs $\W=\{(A_1, B_1), (A_2, B_2), \cdots, (A_s, B_s)\}$ is an $\eps$-WSPD if the following conditions hold: (1) for any $i\leq s$, $A_i\subseteq A$, $B_i\subseteq B$ (2) for each pair of points $(a,b) \in A\times B$, there exists a unique pair $(A_j, B_j)\in \W$ such that $a \in A_j$ and $b\in B_j$ (3) for any $i\leq s$, $\max\{\diam(A_i), \diam(B_i)\} \leq \eps \cdot \dist(A_i, B_i)$, where $\diam(X) = \max_{x,y\in X} \dist(x,y)$ and $\dist(X, Y) = \min_{x \in X, y\in Y} \dist(x,y)$ (see Figure~\ref{fig:wspd}). 
As shown in~\cite{har2011geometric, har2006fast} if $sp(A\cup B)=n^{O(1)}$, a quad tree $T$ on $A\cup B$ can be used to construct, in time $O(n\log n + \eps^{-d}n)$, a WSPD $\W$ of size $O(\eps^{-d}n)$ such that
each pair $(A_i, B_i)\in \W$ is associated with pair of cells $(\square_i, \boxplus_i)$ in $\T$ where $A_i=A\cap \square_i$ and $B_i=B\cap \boxplus_i$.
It is also known that for each pair $(A_i, B_i)\in \W$ (i) $\square_i\cap \boxplus_i=\emptyset$, (ii)
$\max\{\diam(\square_i), \diam(\boxplus_i)\}\leq \eps\dist(\square_i, \boxplus_i)$, and each cell appears in $O(\eps^{-d}\log n)$ cells (see Figure~\ref{fig:wspd}).
We will use $\W=\{(\square_1,\boxplus_i), \ldots, (\square_s,\boxplus_s)\}$ to denote the WSPD, with $A_i, B_i$ being implicitly defined from the cells.
Using the techniques in~\cite{callahan1995dealing, fischer2005dynamic}, the quad tree $\T$ and the WSPD $\W$ can be maintained under insertions and deletions of points in $\O(\eps^{-d})$ time.

\paragraph{Data structure.} 
We construct a quad tree $\T$ on $A\cup B$.
For each node $u \in \T$, we store a pointer $A_u$ (and $B_u$) to the leftmost leaf of subtree $\T_u$ that contains a point from $A$ (and $B$). Furthermore, we store sorted lists $L_A$ and $L_B$ of the leaves that contain points from $A$ and $B$, respectively. We use these pointers and lists to report points in $\square_u$ with $O(1)$ delay.
Using $\T$, we can construct a WSPD $\W=\{(\square_1, \boxplus_1),\ldots, (\square_s, \boxplus_s)\}$, $s=O(\eps^{-d})$.
For each $i$, let $\Delta_i=\min_{p\in \square_i, q\in\boxplus_i}\dist(p,q)$. We store all pairs $(\square_i, \boxplus_i)$ in a red-black tree $\mathcal{Z}$ using $\Delta_i$ as the key. 
The data structure has $O(\eps^{-d}n)$ size and $O(\eps^{-d}n\log n)$ construction time.

\paragraph{Update.} After inserting or deleting an input point, the quad tree $\T$ and $W$ can be updated in $\O(\eps^{-d})$ time, following the standard techniques in ~\cite{callahan1995dealing, fischer2005dynamic}. As there are at most $\O(\eps^{-d})$ pairs changed, we can update the tree $\mathcal{Z}$ in $\O(\eps^{-d})$ time. Furthermore, we note that there are only $O(1)$ changes in the structure of quad tree $\T$ and the height of $\T$ is $O(\log n)$, so we can update all necessary pointers $A_u, B_u$ and sorted lists $L_A, L_B$ in $O(\log n)$ time.





\paragraph{Enumeration.}
Let $r$ be the threshold parameter specified as part of a query.
We traverse the tree $\mathcal{Z}$ in order and report pairs of cells until we reach a pair $(\square_j, \boxplus_j)$ with $\Delta_j>r$.
For each pair $(\square_i, \boxplus_i)$ reported, we traverse we enumerate $(a,b)\in (A\cap \square_i) \times (B\cap \boxplus_i)$ using the stored pointers and the sorted lists $L_A, L_B$.
The delay guarantee is $O(1)$.

Let $(a, b) \in A \times B$ be a pair with $\dist(a,b)\leq r$. Implied by the definition, there exists a unique pair $(A_i, B_i)\in \W$ such that $a\in A_i$ and $b\in B_i$. Notice that  $\dist(\square_i, \boxplus_i) \leq\dist(a,b) \leq r$. Thus, all results of $A_i \times B_i$ will be reported, including $(a,b)$.  Next, let $(\square_i, \boxplus_i)$ be a pair that is reported by the enumeration procedure in $\mathcal{Z}$, with $\dist(\square_i, \boxplus_i)\leq r$.
For any pair of points $x \in \square_i, y \in \boxplus_i$, we have
$\dist(x,y) \le \dist(\square_i,\boxplus_i)+\diam(\square_i)+\diam(\boxplus_i) 
\le  (1+2\cdot \frac{\eps}{2}) \cdot \dist(\square_i,\boxplus_i) \le (1+\eps)r$,
thus $\dist(a,b)\leq (1+\eps)r$ for any pair $(a,b)\in A_i\times B_i$.

\begin{theorem}
Let $A, B$ be two sets of points in $\mathbb{R}^d$ for constant $d$, with $O(n^{O(1)})$ spread and $|A| + |B| = n$.
A data structure of $O(\eps^{-d}n)$ space can be built in $\O(\eps^{-d}n)$ time and updated in $\O(\eps^{-d})$ time, while supporting $\eps$-approximate enumeration for similarity join under any $\ell_p$ metric with $O(1)$ delay, for any query similarity threshold $r$.
\end{theorem}

\subsection{Fixed distance threshold}
\label{sec:grid}

Without loss of generality we assume that $r=1$. We use a grid-based data structure for enumerating similarity join with fixed distance threshold $r$.
%

\newcommand{\nonEmptyCells}{\mathcal{C}_{NE}}

\paragraph{Data structure.} Let $\mathcal{G}$ be an infinite uniform grid\footnote{When extending it to any $\ell_p$ norm, the size of each grid cell is $\eps/(2d^{1/p})$ and the diameter is $\frac{\epsilon}{2}$.} in $\Re^d$, where the size of each grid cell is $\frac{\eps}{2\sqrt{d}}$ and the diameter is $\frac{\eps}{2}$. For a pair of cells $c,c' \in \G$, define $\dist(c,c') = \min_{p \in c, q \in c'} \dist(p,q)$. 
Each grid cell $c\in \mathcal{G}$ is associated with (1) $A_c=A \cap c$; 
(2) $B_c=B \cap c$; 
(3) $m_c=\sum_{c': \dist(c,c')\leq 1} |B_{c'}|$ as the 
number of points in $B$ that lie in a cell $c'$ within distance $1$ from cell $c$.
Let $\nonEmptyCells\subseteq \mathcal{G}$ be the set of all non-empty cells, $\nonEmptyCells=\{c\in \mathcal{G}\mid A_c\cup B_c\neq \emptyset\}$.
A grid cell $c\in \nonEmptyCells$ is {\em active} if and only if $A_c\neq \emptyset$ and $m_c > 0$ (see Figure~\ref{fig:Grid} for an example). Let $\C \subseteq \nonEmptyCells$ be the set of active grid cells (Figure~\ref{fig:Grid}).
Notice that a grid cell is stored when there is at least one point from $A$ or $B$ lying inside it, so $|\nonEmptyCells|\leq n$. Finally, we build a balanced search tree on $\C$ so that whether a cell $c$ is stored in $\C$ can be answered in $O(\log n)$ time. Similarly, we build another balanced search tree to store the set of non-empty cells $\nonEmptyCells$.

\paragraph{Update.}
Assume point $a \in A$ is inserted into cell $c\in \mathcal{G}$. If $c$ is already in $\nonEmptyCells$, simply add $a$ to $A_c$. Otherwise, we add $c$ to $\nonEmptyCells$ with $A_c=\{a\}$ and update $m_c$ as follows. We visit each cell $c' \in \nonEmptyCells$ with $\dist(c,c') \le 1$, and add $|B_{c'}|$ to $m_c$. A point of $A$ is deleted in a similar manner.
Assume point $b \in B$ is inserted into cell $c\in \mathcal{G}$. If $c \notin \nonEmptyCells$, we add it to $\nonEmptyCells$. In any case, we first insert $b$ into $B_c$ and for every cell $c'\in \nonEmptyCells$ with $\dist(c,c') \le 1$, we increase $m_{c'}$ by $1$ and add $c'$ to $\C$ if $c'$ turns from inactive to active.
A point from $B$ is deleted in a similar manner. 
As there are $O(\eps^{-d})$ cells within distance $1$ from $c$, this procedure takes $\O(\eps^{-d})$ time.


\paragraph{Enumeration.} For each active cell $c\in \C$, we visit each cell $c'\in \nonEmptyCells$ within distance $1$. If $B_{c'}\neq \emptyset$, we report all pairs of points in $A_c \times B_{c'}$.  It is obvious that each pair of points is enumerated at most once. For an active cell $c$, there must exists a pair $(a\in A_c, b\in B_{c'})$ for some cell $c' \in \nonEmptyCells$ such that $\dist(a,b) \leq \dist(c,c')+\diam(c)+\diam(c')\leq 1+\eps$. So it takes at most $O(\eps^{-d}\log n)$ time before finding at least one result for $c$, thus the delay is $O(\eps^{-d}\log n)$.  Furthermore, consider every pair of points $a,b$ with $\dist(a,b)\leq 1$. Assume $a\in c$ and $b\in c'$. By definition, $c$ must be an active grid cell. Thus, $(a,b)$ will definitely be enumerated in this procedure, thus guaranteeing the correctness of $\eps$-enumeration.

\begin{theorem}
Let $A, B$ be two sets of points in $\mathbb{R}^d$ for some constant $d$, with $|A| + |B| = n$.
A data structure of $O(n)$ size can be constructed in $O(n\eps^{-d}\log n)$ time and updated in $O(\eps^{-d}\log n)$ time, while supporting $\eps$-approximate enumeration of similarity join under any $\ell_p$ metric with $O(\eps^{-d}\log n)$ delay.
\end{theorem}

Note that if for each active cell $c\in \C$, we store the cells within distance $1$ that contain at least a point from $B$, i.e., $\{c'\in C\mid \dist(c,c')\leq 1, B_c\neq \emptyset\}$, then the delay can be further reduced to $O(1)$ but the space becomes $O(\eps^{-d}n)$.



%% file: lshMain.tex
\section{Similarity Join in High Dimensions}
\label{sec:highd}

So far, we have treated the dimension $d$ as a constant.
In this section we describe a data structure for approximate similarity join using the \emph{locality sensitive hashing} (LSH) technique, so that the dependency on $d$ is a small polynomial. For simplicity, we assume that $r$ is fixed, however our results can be extended to the case in which $r$ is part of the enumeration query.

For $\eps > 0$, $0<p_2<p_1\leq 1$, a family $\H$ of hash functions is $(r, (1+\eps)r, p_1, p_2)$-sensitive, if for any uniformly chosen hash function $h \in H$ and any two points $x,y$: (\romannumeral 1) $\Pr[h(x) = h(y)] \ge p_1$ if
$\dist(x, y) \le r$; and (\romannumeral 2) $\Pr[h(x) = h(y)] \le p_2$ if $\dist(x, y) \ge
(1+\eps)r$. The quality of $\H$ is measured by $\rho= \frac{\ln p_1}{\ln p_2} <1$,
which is upper bounded by a number that depends only on $\eps$; and $\rho = \frac{1}{1+\eps}$ for many common distance functions~\cite{gionis1999similarity,datar2004locality,har2012approximate}. For $\ell_2$ the best result is $\rho\leq \frac{1}{(1+\eps)^2}+o(1)$
~\cite{andoni2006near}. 

The essence of LSH is to hash ``similar'' points into the same buckets with high probability.
A simple approach based on LSH is to (\romannumeral 1) hash points into buckets; (\romannumeral 2) probe each bucket and check for each pair of points $(a,b) \in A \times B$ inside the same bucket whether $\dist(a,b) \le r$; and (\romannumeral 3) report $(a,b)$ if the inequalities holds. 
However, two challenges arise for enumeration. First, without any knowledge of false positive results inside each bucket, checking every pair of points could lead to a huge delay. Our key insight is that after checking a specific number of pairs of points in one bucket (this number will be determined later), we can safely skip the bucket, since any pair of result missed in this bucket will be found in another one with high probability.
Second, one pair of points may collide under multiple hash functions, so an additional step is necessary in the enumeration to remove duplicates. If we wish to keep the size of data structure to be near-linear and if we are not allowed to store the reported pairs (so that the size remains near linear), detecting duplicates requires some care.

Our data structure and algorithm use a parameter $M$, whose value will be determined later. Since we do not define new hash functions, our results hold for any metric for which LSH works, in particular for Hamming, $\ell_2, \ell_1$ metrics.

\newcommand{\buckets}{\Xi}

\paragraph{Data Structure.} We fix an LSH family $\H$. Let $\rho$ be its quality parameter. We randomly choose $\tau = O(n^{\rho})$ hash functions. Let $\buckets$ be the set of buckets over all hash functions, each corresponding to one possible value in the range of hash functions. We maintain some extra statistics for buckets in $\buckets$. For a bucket $\square$, let $A_\square = A \cap \square$ and $B_\square = B \cap \square$.
We choose two arbitrary subsets $\AM, \BM$
of $A_\square, B_\square$, respectively, of $M$ points each. We choose $M=O(n^{\rho})$.
For each point $a\in \AM$, we maintain a counter $\beta_a=|\{b\in \BM\mid \dist(a,b)\leq 2(1+\eps)r\}|$, i.e., the number of points in $\BM$ with distance at most $2(1+\eps)r$ from $a$.
We store $\AM$ in an increasing order of their $\beta$ values. If there exists $a\in \AM$ with $\beta_a>0$, we call $\square$ {\em active} and store an arbitrary pair $(a,b)\in \AM \times \BM$ with $\dist(a,b)\leq 2(1+\eps)r$ as its \emph{representative} pair, denoted as $(a_\square, b_\square)$. 
Let $\C$ denote the set of active buckets.
To ensure high-probability guarantee, 
we maintain $O(\log n)$ copies of this data structure.

\medskip
Before diving into the details of update and enumeration, we give some intuition about active buckets. 
Given a set $P$ of points and a distance threshold $r$, let $\overline{\mathcal{B}}(q, P, r) = \{p \in P\mid \dist(p,q) >r\}$. 
For any pair of points $(a,b) \in A \times B$ and a hashing bucket $\square$, we refer $\square$ as the {\em proxy bucket} for $(a,b)$ if (\romannumeral 1) $a \in A_\square, b \in B_\square$; (\romannumeral 2) $|\bar{\mathcal{B}}(a, A_\square \cup B_\square, (1+\eps)r)| \le M$. A crucial property of proxy bucket is captured by Lemma~\ref{lem:proxy-2:Main} (the proof will be given later),  which implies that it is safe (with high probability) to skip a bucket after we have seen up to $M^2$ far away pairs of points inside, since the truly similar pairs of points inside it will be reported from other buckets. In this way, we only need to enumerate join results from active buckets. 

\begin{lemma}
\label{lem:proxy-2:Main}
   For any bucket $\square$, if there exist $M$ points from $A_\square$ and $B_\square$ each, such that none of the $M^2$ pairs has its distance within $2(1+\eps)r$, $\square$ is not a proxy bucket for any pair $(a, b) \in A_\square \times B_\square$ with $\dist(a,b) \le r$. 
\end{lemma}

\paragraph{Update.} When a point is inserted, say $a\in A$, we visit every bucket $\square$ into which $a$ is hashed and insert $a$ to $A_\square$. If $|\AM|\geq M$, we do nothing. If $|\AM|<M$, we insert $a$ in $\AM$ and compute its counter $\beta_a$ by visiting all points in $\BM$. If $\beta_a>0$ and $\square\notin \C$, we add $\square$ to $\C$ and store the representative pair of $\square$ defined as $(a,b)$, where $b\in \BM$ is an arbitrary point of $\BM$ with $\dist(a,b)\leq 2(1+\eps)r$. Notice that there always exists such a point $b$ because $\beta_a>0$.
When a point is deleted, say $a\in A$, we visit every bucket $\square$ into which $a$ is hashed and delete $a$ from $A_\square$.
If $a\in \AM$, we delete it from $\AM$ and insert an arbitrary point (if any) from  $A_\square\setminus \AM$ into $\AM$.
If $a = a_\square$, i.e., $a$ participates in the representative pair of $\square$, we find a new representative pair by considering an arbitrary point $a'\in \AM$ with $\beta_{a'}>0$. If no such point exists, we remove $\square$ from $\C$.
The insertion or deletion of a point $b \in B$ is similar. 
After performing $n/2$ updates, we reconstruct the entire data structure from scratch.

\paragraph{Enumeration.} The high-level idea is to enumerate the representative pair from every active bucket and recompute new representative pairs for it. Assume a representative pair $(a,b)$ is found in a bucket $\square \in \C$. Next, we enumerate all pairs that involve the point $a$.

For any bucket $\square'$ such that $a
\in A_{\square'}$, let $X(\square', a)\subseteq B_\square$ be the set of \emph{marked} points of $B_\square$ that the enumeration procedure has already reported their pairs with $a$. For example, if $(a,b)$ is reported and both $a,b$ lie in a set of buckets $C$, then $b\in X(\square', a)$ for each $\square'\in C$.

Let $\C(a) \subseteq \C$ be the set of active buckets containing $a$.
We visit every bucket $\square \in \C(a)$, and check the distances between $a$ and points in $B_\square\setminus X(\square, a)$. Each time a pair $(a,b)$ with $\dist(a,b) \le 2(1+\eps)r$ is found, we report it and invoke a de-duplication step on $(a,b)$ to make sure that we will not report $(a,b)$ again. Details of the de-duplication procedure will be given later. When more than $M$ points from $B_\square$ have been checked without finding a pair with distance less than $2(1+\eps)r$ (or if all all points in $B_\square$ have been considered), 
we remove\footnote{In the enumeration, ``remove'' means ``conceptually mark'' instead of changing the data structure itself.} $a$ from $A_\square$, remove $\square$ from $\C(a)$, and skip this bucket.
If $a\in \AM$ \footnote{For simplicity, at the beginning of the enumeration procedure we can construct a copy of $\AM$ of each bucket $\square$ so that the original points in $\AM$ remain the same when the next enumeration query is executed. This does not affect the asymptotic complexity of the delay guarantee.} we remove $a$ from $\AM$ and we insert another point from $A_\square\setminus \AM$ that we have not visited before so that the next representative pair of $\square$ (if any) does not contain any point from $A$ that we have already visited in the current numeration phase.
Once all buckets in $\C(a)$ have been visited, we can just pick an arbitrary active bucket in $\C$ with its representative pair $(a',b')$ (it will always be the case that $a' \neq a$), and start the enumeration for $a'$.  

Finally, we avoid reporting a pair more than once, as follows. Once a pair $(a,b)$ is enumerated, we go over each bucket $\square$ into which both $a,b$ are hashed, and mark $b$ with $X(\square, a)$ to avoid further repeated enumeration. Moreover, if $(a,b)$ is the representative pair of $\square$, we check at most $M$ points from $B_\square$, whether there exists $b'\in B_\square$ such that $\dist(a,b')\leq 2(1+\eps)r$.
If such a pair exists, we store it as the new representative pair (with respect to $a$) for $\square$. Otherwise, we remove $a$ from $A_\square$, remove $\square$ from $\C(a)$, and if $a\in \AM$ update $\AM$ accordingly.

\paragraph{Correctness analysis.} The de-duplication procedure guarantees that each pair of points is enumerated at most once. It remains to show that $(1+2\eps)$-approximate enumeration is supported. To prove it, we first point out some important properties of our data structures.

\begin{proof}[Proof of Lemma~\ref{lem:proxy-2:Main}]
Let $A',B'$ be two sets of $M$ points from $A_\square, B_\square$ respectively. We assume that all pairs of points in $A' \times B'$ have their distances larger than $2(1+\eps)r$.
Observe that $\square$ is not a proxy bucket for any pair $(a\in A', b\in B')$.
It remains to show that $\square$ is not a proxy bucket for any pair $(a\in A_\square\setminus A', b\in B_\square)$.
%
%
Assume $b\in B_\square \setminus B'$ (the case is similar if $b\in B'$). 
If $A' \subseteq \bar{\mathcal{B}}(a,A,(1+\eps)r)$ or $B' \subseteq \bar{\mathcal{B}}(a,B,(1+\eps)r)$,
$\square$ is not a proxy bucket for $(a,b)$. Otherwise, there must exist at least one point $a' \in A'$ as well as $b' \in B'$ such that $\dist(a,a') \le (1+\eps)r$ and $\dist(a,b') \le (1+\eps)r$, so $\dist(a',b') \le \dist(a,a') + \dist(a,b') \le 2(1+\eps)r$.
Thus, $(a', b') \in A' \times B'$ is a pair within distance $2(1+\eps)r$, coming to a contradiction. 
\end{proof}

We show that $(1+2\eps)$-approximate enumeration is supported with probability $1-1/n$.
It can be easily checked that any pair of points farther than $2(1+\eps)r$ will not be enumerated.
Hence, it suffices to show that all pairs within distance $r$ are enumerated with high probability.
From~\cite{gionis1999similarity, har2011geometric, indyk1998approximate} it holds that for $M=O(n^\rho)$, any pair $(a,b)$ with $\dist(a,b)\leq 1$ has a proxy bucket with probability $1-1/n$.
Let $\square$ be a proxy bucket for pair $(a,b)$.
Implied by Lemma~\ref{lem:proxy-2:Main}, there exist no $M$ points from $A_\square$ (for example $\AM$) and $M$ points from $B_\square$ (for example $\BM$) such that all $M^2$ pairs have their distance larger than $2(1+\eps)r$, so $\square$ is active. Moreover, there exist no $M$ points from $B_\square$ such that all of them have distance at least $2(1+\eps)r$ from $a$, so $\square$ is an active bucket for $a$. Hence, our enumeration algorithm will report $(a,b)$.

%

\paragraph{Complexity analysis.}
Recall that $\tau, M=O(n^\rho)$.
The data structure uses $O(dn+ n\tau\log n)$ space since we only use linear space with respect to the points in each bucket. The update time is $\O(dM \cdot \tau)$ as there are $\O(\tau)$ buckets to be investigated and it takes $\O(dM)$ time to update the representative pair. After $n/2$ updates we re-build the data structure so the update time is amortized. The delay is $\O(dM\cdot \tau)$; consider the enumeration for point $a$. It takes $\O(dM \cdot \tau)$ time for checking all buckets while $\O(dM \cdot \tau)$ time for de-duplication. 
%
%
%

Alternatively, we can insert or delete points from $A\cup B$ without maintaining the sets $\AM, \BM$ for every bucket $\square$. In the enumeration phase, given a bucket $\square$, we can visit $M$ arbitrary points from $A_\square$ and $M$ arbitrary points from $B_\square$ and compute their pairwise distances. If we find no pair $(a\in A_\square, b\in B_\square)$ with $\dist(a,b)\leq 2(1+\eps)r$ then we skip this bucket. Otherwise we report the pair $(a,b)$ and we run the de-duplicate procedure. In this case the update time is $\O(dn^\rho)$ but the delay is $O(dn^{3\rho})$.
We conclude the following result:
\begin{theorem}
\label{the:high-non-uniform:Main}
Let $A$ and $B$ be two sets of points in $\mathbb{R}^d$, where $|A|+|B|=n$ and let $\eps, r$ be positive parameters. For $\rho=\frac{1}{(1+\eps)^2}+o(1)$,
a data structure of $\O(dn+n^{1+\rho})$ size can be constructed in $\O(dn^{1+2\rho})$ time, and updated in $\O(dn^{2\rho})$ amortized time, while supporting $(1+2\eps)$-approximate enumeration for similarity join under the $\ell_2$ metric with $\O(dn^{2\rho})$ delay. Alternatively, a data structure of $\O(dn+n^{1+\rho})$ size can be constructed in $\O(dn^{1+\rho})$ time, and updated in $\O(dn^{\rho})$ amortized time, while supporting $(1+2\eps)$-approximate enumeration with $\O(dn^{3\rho})$ delay.
\end{theorem}
The same result holds for Hamming and $\ell_1$ metrics with $\rho=\frac{1}{1+\eps}$.
Using~\cite{indyk1998approximate}, for the Hamming metric and $\eps>1$ we can get $M=O(1)$.
Skipping the details, we have:
\begin{theorem}
\label{the:Shigh-non-uniform:Main}
Let $A$ and $B$ be two sets of points in $\mathbb{H}^d$, where $|A|+B|=n$ and let $\eps, r$ be positive parameters. For $\rho=\frac{1}{1+\eps}$,
a data structure of $\O(dn+n^{1+\rho})$ size can be built in $\O(dn^{1+\rho})$ time, and updated in $\O(dn^{\rho})$ amortized time, while supporting $(3+2\eps)$-approximate enumeration for similarity join under the Hamming metric with $\O(dn^{\rho})$ delay.
\end{theorem}

In Appendix~\ref{sec:lsh}, we show the full description of algorithms and proofs. We also show that our results can be extended to the case where $r$ is part of the enumeration procedure.
Finally, we show a lower bound relating similarity join to the approximate nearest neighbor query.

%% file: appendix.tex
\appendix

\section{Triangle Similarity join}
\label{sec:trianglejoin}
In this section we propose data structures for the approximate triangle join queries. Our results can be extended to $m$-clique join queries, for constant $m$. For simplicity we describe the results for the triangle join $m=3$.
Let $A, B, S\in \Re^d$ be three sets of points such that $|A|+|B|+|S|=n$.
We first consider the when the distance threshold $r$ is fixed and then lift this assumption.

\subsection{Fixed distance threshold}
\label{appndx:approxtriangle}
The data structure we construct works for any $\ell_p$ norm. For simplicity, we describe it for $\ell_2$ first and extend it to any $\ell_p$ metric at last.
In this subsection, we use $\dist(a,b)=||a-b||_2$.

As we had in Section~\ref{sec:grid} let $\mathcal{G}$ be an infinite uniform grid in $\Re^d$ where the size of each grid cell is $\eps/(2\sqrt{d})$, so its diameter is $\eps/2$ (using the $\ell_2$ distance).
For each grid cell $c\in \mathcal{G}$ we store $A_c=A\cap c$, $B_c=B\cap c$, and $S_c=S\cap c$. Furthermore, we store a counter $m_c=|\{(b,s)\in B\times S\mid \exists c_1, c_2 \in \mathcal{G} \text{ s.t. } b\in c_1, s\in c_2, \dist(c_1,c_2)\leq 1, \dist(c_1,c)\leq 1, \dist(c_2,c)\leq 1\}|$, i.e., the number of pairs $(b,s)$ whose cells along with cell $c$ are within distance $1$. Let $\nonEmptyCells$ be the non-empty cells, i.e., $\nonEmptyCells=\{c\in \mathcal{G}\mid A_c\cup B_C\cup S_c\}\neq \emptyset$. A grid cell $c\in C$ is active if and only if $A_c\neq \emptyset$ and $m_c>0$. Let $\C\subseteq \nonEmptyCells$ be the set of active grid cells. We construct a balanced search tree to answer efficiently if a cell is already in $\C$. Similarly we create a balanced search tree for the cells in $\nonEmptyCells$.
Our data structure has $O(n)$ space.

We first describe the updates. 
Assume that we insert a point $a\in A$. If $a$ lies in a cell $c\in \nonEmptyCells$ then we insert $a$ in $A_c$. If $a$ is inserted to a cell that did not exist then we create $c$, we add it in $\nonEmptyCells$ and we set $A_c=\{a\}$. Then we need to find the value $m_c$. The algorithm visits all existed cells around $c\in \nonEmptyCells$ within distance $1$. Let $c_1$ be such a cell such that $B_{c_1}\neq \emptyset$ or $S_{c_1}\neq \emptyset$. We need to count all points in $B$ and $S$ that lie in cells within distance $1$ from both $c$ and $c_1$.
Notice that these cells are inside a rectangle $R$. Indeed, if $R_1$ is a square of radius $1$ around $c$ and $R_2$ is a square of radius $1$ around $c_1$ then $R=R_1\cap R_2$ is a rectangle.
We visit all grid cells inside $R$ and find the number of points from $B, S$ in $R$. Let $m_B=|B\cap R|$ and $m_S=|S\cap R|$. We update $m_c$ with $m_c+|B_{c_1}|\cdot m_S+|S_{c_1}|\cdot m_B$. In the end it is easy to verify that $m_c$ has the correct value.
Next, assume that we remove a point $a\in A$. Let $c$ be the cell of point $a$. We remove $a$ from $A_c$ and if $A_c=\emptyset$
and $c\in \C$ then we remove $c$ from $\C$. If $A_c=B_c=S_c=\emptyset$ we remove $c$ from $\nonEmptyCells$. Since there are $O(\eps^{-d})$ grid cells in a square of radius $1$ we need
$O(\eps^{-2d}\log n)$ time to insert $a$ and $O(\eps^{-d}\log n)$ to remove $a$.
Then we continue by updating a point $b\in B$ (the method is similar to update $s\in S$).
Assume that we add $b\in B$ in a cell $c$ (if $c$ did not exist we create it) and we insert it in $B_c$. The goal is to update all counters $m_{c_1}$ within distance $1$ from $c$.
We start by visiting all cells $c_1\in \nonEmptyCells$ within distance $1$ from $c$. We need to update the value of $m_{c_1}$. In particular, we need to count the number of points in $S$ that lie in cells within distance $1$ from both $c, c_1$. This is similar to what we had for the insertion of $a$ so we can count it
by visiting all grid cells within distance $1$ from $c, c_1$. Let $m_S$ be the result. Then, we update $m_{c_1}\leftarrow m_{c_1}+m_S$. Finally, assume that we remove a point $b\in B$ from a cell $c$. We remove $b$ from $B_c$ and again, we need to visit all cells $c_1$ within distance $1$ and update their $m_{c_1}$ values by $m_{c_1}\leftarrow m_{c_1}-m_S$ ($m_S$ can be found as we explain in the previous case). If $c_1\in \C$ and $m_{c_1}=0$ we remove $c_1$ from $\C$. In the end, if $A_c=B_c=S_c=\emptyset$ we remove $c$ from $\nonEmptyCells$. Again, it is easy to observe that $m_c$ have the correct values for all $c\in \nonEmptyCells$ and hence $\C$ is the correct set of active cells.
We need $O(\eps^{-2d}\log n)$ time to insert or remove a point in $B$.

Next, we describe the enumeration procedure. For each $c\in \C$ we consider every $a\in A_c$. We visit each cell $c_1\in \nonEmptyCells$ around $c$ within distance $1$. Then we visit each cell $c_2\in \nonEmptyCells$ within distance $1$ from both $c_1, c$. We report (if any) the points $a\times B_{c_1}\times S_{c_2}$ and $a\times S_{c_1}\times B_{c_2}$.
We show the correctness of our method. Let $(a\in A, b\in B, s\in S)$ be a triad within distance $1$. Let $a\in c_1, b\in c_2, s\in c_3$ for $c_1, c_2, c_3\in \nonEmptyCells$. Notice that $\dist(c_1,c_2),\dist(c_1,c_3), \dist(c_2,c_3)\leq 1$. From the update procedure we have that $m_{c_1}>0$, hence, $c_1\in \C$. The algorithm will visit $c_1$ and it will also consider $c_2$ since $\dist(c_1,c_2)\leq 1$. Then it will also consider $c_3$ since $\dist(c_1,c_3)\leq 1$ and $\dist(c_2,c_3)\leq 1$. Hence our enumeration procedure will return the triad $(a,b,s)$. Furthermore, it is straightforward to see that i) our enumeration algorithm will never report a triad $(a,b,s)$ such that a pairwise distance is greater than $1+\eps$, and ii) whenever $c\in \C$ there will always be a triad $(a\in A_c, b\in B, s\in S)$ to report. Finally, since our enumeration algorithm reports points that lie in cells with pairwise distance $1$ it might be possible that it will return $(a,b,s)$ such that $\dist(a,b)\leq \dist(c_1,c_2)+\diam(c_1)+\diam(c_2)\leq 1+\eps$, $\dist(a,c)\leq 1+\eps$, and $\dist(b,c)\leq 1+\eps$.
The delay is $O(\eps^{-2d}\log n)$.

The same result can be extended to any $\ell_p$ norm by considering grid cells of size $\eps/(2d^{1/p})$.

\begin{theorem}
Let $A, B, S$ be three sets of points in $\mathbb{R}^d$, with $|A| + |B| +|S| = n$.
A data structure of $O(n)$ space can be constructed in $O(n\eps^{-2d}\log n)$ time and updated in $O(\eps^{-2d}\log n)$ time, while supporting $\eps$-approximate enumeration of triangle similarity join queries under any $\ell_p$ metric with $O(\eps^{-2d}\log n)$ delay.
\end{theorem}

For $\ell_1, \ell_\infty$, we can slightly improve the result using a data structure to find $m_B, m_S$ more efficiently. Skipping the details, we can obtain a data structure of $O(n\log^{d-1} n)$ space that can be built in $O(n\log^{d-1} n + n\cdot\min\{\eps^{-d}\log^{d-1} n, \eps^{-2d}\}\log n)$ time and updated in $O(\min\{\eps^{-d}\log^{d-1} n, \eps^{-2d}\}\log n)$ time, while supporting $\eps$-approximate enumeration of triangle similarity join under $\ell_1/\ell_\infty$ metrics with $O(\min\{\eps^{-d}\log^{d-1} n, \eps^{-2d}\}\log n)$ delay.


\subsection{Variable distance threshold}
We describe two data structures for this case. One is based on grid using $O(\eps^{-1}n\log (n))$ space and the other based on WSPD using $O(\eps^{-2d}n)$ space.

\paragraph{Grid-based data structure.}
Assume that the spread $sp(A\cup B\cup S)=n^{O(1)}$ and that all points lie in a box with diagonal length $R$.
The high level idea is to build multiple grids as described in Appendix~\ref{appndx:approxtriangle}. Recall that for each cell $c\in C$, we need to store counters $A_c, B_c, S_c$ and $m_c$. However, the definition of $m_c$ depends on the threshold $r$ which is not known upfront in this case.
Hence we consider multiple thresholds $r_i$.
In particular for each $i\in[0, \log_{1+\eps/4}sp(A\cup B\cup S)]$ we construct a grid for $r_i=\frac{R}{sp(A\cup B\cup S)}(1+\eps/4)^i$ as in Appendix~\ref{appndx:approxtriangle}.
Hence for each $i$ we  maintain the counter $m_c^i$ defined as
$m_c^i=|\{(b,s)\in B\times S\mid \exists c_1, c_2 \in \mathcal{G} \text{ s.t. } b\in c_1, s\in c_2, \dist(c_1,c_2)\leq r_i, \dist(c_1,c)\leq r_i, \dist(c_2,c)\leq r_i\}|$, \footnote{For each $i$ we scale everything so that $r_i=1$, as we did in Appendix~\ref{appndx:approxtriangle}.} and the set of active cells $\C_i$.
Notice that there are $O(\eps^{-1}\log n)$ different values of $i$.
For a point insertion or deletion the algorithm updates all necessary counters $m_c^i$ and active cells $\C_i$ for all $i$. For an enumeration query, assume that $r$ is the query threshold. Notice that $\frac{R}{sp(A\cup B\cup S)}\leq r\leq R$, otherwise the result is trivial. Running a binary search on the values of $i$ we find the smallest $i$ such that $r\leq r_i$. Then using only the active cells $\C_i$ and the counters $m_c^i$ we enumerate all triangles within distance $r_i$. The delay guarantee is the same as in Appendix~\ref{appndx:approxtriangle}, $O(\eps^{-2d}\log n)$. We conclude with the next theorem.

\begin{theorem}
Let $A, B, S$ be three sets of points in $\mathbb{R}^d$ for constant $d$, with $O(\poly(n))$ spread, $|A| + |B| +|S| = n$, where $A\cup B\cup S$ lie in a hyper-rectangle with diagonal length $R$.
A data structure of $O(\eps^{-1}n\log n)$ space can be constructed in $O(n\eps^{-2d-1}\log^2 n)$ time and updated in $O(\eps^{-2d-1}\log^2 n)$ time, while supporting $\eps$-approximate enumeration of triangle similarity join queries under any $\ell_p$ metric with $O(\eps^{-2d}\log n)$ delay, for any query distance threshold $r$.
\end{theorem}

\paragraph{WSPD-based data structure.}
We describe the main idea here. Assume that $sp(A\cup B\cup S)=n^{O(1)}$.
Let $\mathcal{W}_{A,B}$ be the WPSD construction of $A, B$ as in Section~\ref{sec:wspd}. Similarly, we consider $\mathcal{W}_{A,S}$, and $\mathcal{W}_{B,S}$. For each pair $(A_i, B_i)\in \mathcal{W}_{A,B}$, let $\dist(\square_i, \boxplus_i)=r_i$.
Let $c_i$ be the center of $\square_i$ and $c_i'$ be the center of $\boxplus_i$.
Let $\mathcal{L}_i$ be the lune
(intersection) of the spheres with radius $r_i$ and with centers $c_i, c_i'$. We run a query with $\mathcal{L}_i$ on a quadtree $\T_S$ on the points $S$. We get $O(\eps^{-d})$ quadtree boxes. Then we construct the triplets $\mathcal{W}_{A,B}'=\{(A_1,B_1,S_1),\ldots, (A_{\xi}, B_{\xi}, S_{\xi})\}$, where $\xi=O(\eps^{-2d}n)$.
Similarly, we construct $\mathcal{W}_{A,S}', \mathcal{W}_{B,S}'$.
Let $\mathcal{W}'=\mathcal{W}_{A,B}'\cup \mathcal{W}_{A,S}'\cup \mathcal{W}_{B,S}'$.
We can show that each triplet $(a,b,s)\in A\times B \times S$ can be found in at least one triplet $(A_i, B_i, S_i)$ in $\mathcal{W}'$.
In particular, let $(a,b,s)\in A\times B\times S$ be a triplet such that (without loss of generality) $\dist(a,b)\geq \dist(a,s)\geq \dist(b,s)$. From the definition of the WPSD $\mathcal{W}_{A,B}$ we have that there exists a unique pair $(A_i, B_i)$ such that $a\in A_i$ and $b\in B_i$. Notice that $\dist(a,s), \dist(b,s)<\dist(a,b)$ so $s$ should lie in the lune $\mathcal{L}_i$ so there exists a triplet $(A_i, B_i, S_i)\in \mathcal{W}_{A,B}'\subseteq \mathcal{W}'$ such that $a\in A_i, b\in B_i, s\in S_i$.
In addition, due to the bounded spread we have that each node participates in at most $O(\eps^{-2d}\log n)$ triplets in $\mathcal{W}'$ and each point belongs in at most $O(\eps^{-2d}\log^2 n)$ triplets in $\mathcal{W}'$. Hence, each update takes $\O(\eps^{-2d})$ time. Using a tree $\mathcal{Z}$ as in Section~\ref{sec:wspd} and following a deduplication method as in Section~\ref{sec:highd} we can execute $\eps$-approximate enumeration of all triplets $(a,b,s)$ within distance $r$ in $\O(\eps^{-2d})$ delay.

\begin{theorem}
Let $A, B, S$ be three sets of points in $\mathbb{R}^d$ for constant $d$, with $O(\poly(n))$ spread and $|A| + |B| = n$.
A data structure of $O(\eps^{-2d}n)$ space can be built in $\O(\eps^{-2d}n)$ time and updated in $\O(\eps^{-2d})$ time, while supporting $\eps$-approximate enumeration for similarity join under any $\ell_p$ metric with $\O(\eps^{-2d})$ delay, for any query distance threshold $r$.
\end{theorem}

\input{lsh}

%% file: lsh.tex
\section{Similarity Join in High Dimensions}
\label{sec:lsh}

So far, we have treated the dimension $d$ as a constant. 
In this section we describe a data structure for approximate similarity join using the \emph{locality sensitive hashing} (LSH) technique so that the dependency on the dimension is a small polynomial in $d$, by removing the exponent dependency on $d$ from the hidden poly-log factor.
For simplicity, we describe our data structure assuming that $r$ is fixed, and in the end we extend it to the case where $r$ is also part of the similarity join query.

For $\eps > 0$, $1\geq p_1 > p_2>0$,
recall that a family $H$ of hash functions is $(r, (1+\eps)r, p_1, p_2)$-sensitive, if for any uniformly chosen hash function $h \in H$, and any two points $x$, $y$, we have (1) $\Pr[h(x) = h(y)] \ge p_1$ if
$\dist(x, y) \le r$; and (2) $\Pr[h(x) = h(y)] \le p_2$ if $\dist(x, y) \ge
(1+\eps)r$. The quality of a hash function family is measured by $ \rho = \frac{\ln p_1}{\ln p_2} < 1$, which is upper bounded by a number that depends only on $\eps$; and $\rho = \frac{1}{1+\eps}$ for many common distance functions~\cite{gionis1999similarity,andoni2006near,datar2004locality,har2012approximate}.
For $\ell_2$ the best result is $\rho\leq \frac{1}{(1+\eps)^2}+o(1)$
~\cite{andoni2006near}.

The essence of LSH is to hash ``similar'' points in $P$ into the same buckets with high probability.
A simple approach to use LSH for similarity join is to (\romannumeral 1) hash points into buckets; (\romannumeral 2) probe each bucket and check, for each pair of points $(a,b) \in A \times B$ inside the same bucket, whether $\dist(a,b) \le r$; and (\romannumeral 3) report $(a,b)$ if the inequalities holds. 

 
 However, two challenges arise with this approach. First, without any knowledge of false positive results inside each bucket, checking every pair of points could lead to a huge delay. Our key insight is that after checking specific number pairs of points in one bucket (this number will be determined later), we can safely skip the bucket, since any pair of result missed in this bucket will be found in another one with high probability.
Secondly, one pair of points may collide under multiple hash functions, so an additional step is necessary in the enumeration to remove duplicates. If we wish to keep the size of data structure to be near-linear and are not allowed to store the reported pairs, detecting duplicates requires some care.

As a warm-up exercise to gain intuition, we first present a relatively easy special case in which input points as well as points inserted are chosen from the universal domain uniformly. In the following, we focus on the general case without any assumption on the input distribution. Our data structure and algorithm use a parameter $M$, whose value will be determined later. Since we do not define new hash functions, all results presented in this section hold for all Hamming, $\ell_2, \ell_1$ metrics.


\subsection{With Uniform Assumption} 
\label{appendix:uniform}
 Under this strong assumption, the LSH technique can be used with a slight modification.  We adopt a LSH family $\H$ with quality parameter $\rho$ and randomly choose $\tau = O(n^{\rho})$ hash functions $g_1, g_2,\cdots, g_\tau$. To ensure our high-probability guarantee (as shown later), we maintain $O(\log n)$ copies of this data structure.

\paragraph{Data structure.} 
Let $C$ be the set of all buckets over all $\tau$ hash functions. 
For each bucket $\square$, let $A_\square, B_\square$ be the set of points from $A, B$ falling into bucket $\square$, respectively. A nice property on $A_\square$ and $B_\square$ is stated in the following lemma, which is directly followed by the balls-into-bins result. 

\begin{lemma}
\label{lem:ball-into-bin}
If input points are randomly and uniformly chosen from the domain universe, 
with probability at least $1 - \frac{1}{n}$, every bucket receives $O(\log n/\log \log n)$ points.
\end{lemma}

As the number of points colliding in each bucket can be bounded by $O(\log n)$, it is affordable to check all pairs of points inside one bucket in $O(\log^2 n)$ time, thus resolving the challenge (1).  Moreover, we introduce a variable $\square_\OUT$ for each bucket $\square \in C$ indicating the number of the pair of tuples within distance $r$ colliding inside $\square$. Obviously, a bucket $\square$ is {\em active} if $\square_\OUT > 0$, and {\em inactive} otherwise. All active buckets are maintained in $\C \subseteq C$, in increasing order of the index of the hash function it comes from.


\paragraph{Update.} 
Assume one point $a \in A$ is inserted. We visit each hash bucket $\square$ into which $a$ is hashed. We insert $a$ into $A_\square$, count the number of pair of points $(a,b \in B_\square)$ with $\dist(a,b) \le r$, and add this quantity to $\square_\OUT$. The case of deletion can be handled similarly.

\paragraph{Enumeration.} Assume $(a,b)$ is to be reported. We check whether $a,b$ have ever collided into any bucket previously. If there exists no index $j < i$ such that $g_j(a) = g_j(b)$, we report it. Then, we need to notify every bucket which also witnesses $(a,b)$ but comes after $\square$. More specifically, for every $j> i$, if $g_j(a) = g_j(b)$ in bucket $\square'$, we decrease $\square'_\OUT$ by $1$, and remove $\square'$ from $\C$ if $\square'_\OUT$ becomes $0$. The pseudocode is given below.

\begin{algorithm}
\caption{{\sc UniEnumLSH}}
All buckets in $\C$ are sorted by the index of hash functions\;
\ForEach{$\square \in \C$}{
    \ForEach{$(a,b) \in \square_A \times \square_B$}{
        \If{$\dist(a,b) \le r$}{
            \textrm{flag} = \textrm{true}\;
            \ForEach{$j \in \{1,2,\cdots,i-1\}$}{
                \If{$g_j(a) = g_j(b)$}{
                    flag $ = \false$\;
                }
            }
            \If{\textrm{flag} = \textrm{true}}{
                {\sc Emit} $(a,b)$\;
                \ForEach{$j \in \{i+1, i+2, \cdots, \tau\}$}{
                    \If{$g_j(a) = g_j(b)$ in $\square$}{
                        $\square_\OUT \gets \square_\OUT -1$\;
                        \If{$\square_\OUT = 0$}{
                        $\C \gets \C - \{\square\}$\;
                        }
                    }
                }
            }
        }
    }
}
\end{algorithm}

\begin{theorem}
\label{thm:uniform-property}
Let $A, B$ be two sets of points in $\mathbb{R}^d$, with $|A| + |B| = n$, and $\eps, r$ be positive parameters. Under uniform assumption, a data structure of $\O(nd)$ size can be constructed in $\O(nd)$ time and updated in $\O(d)$ time, while with probability $1-2/n$ supporting exact enumeration of similarity join with $\O(d)$ delay.
\end{theorem}

\begin{proof}[Proof of Theorem~\ref{thm:uniform-property}]
We first prove the correctness of the algorithm. It can be easily checked that any pair of points with their distance larger than $r$ will not be emitted. Consider 
any pair of points $(a,b)$ within distance $r$. Let $i$ be the smallest index such that $g_i(a) = g_i(b)$ in bucket $\square$. In the algorithm, $(a,b)$ will be  reported by $\square$ and not by any bucket later. Thus, each join result will be enumerated at most once without duplication.
 
In the case of hamming distance, we have $k = \log_2 n$ and $p_1^k = (1 - \frac{r}{d})^{\log n} \in [1/e, 1]$ since $d/r > \log n$ by padding some zeros at the end of all points\footnote{Similar assumption was made in the original paper~\cite{gionis1999similarity} of nearest neighbor search in Hamming distance.}, thus $\tau = 3 \cdot 1/p_1^k  \cdot \ln n = \O(1)$. 

We next analyze the complexity of our data structure. It can be built in $O(nk\tau)$ time with $O(nk\tau)$ space, since there are $n$ vertices in $A \cup B$, at most $O(n\tau)$ non-empty buckets in $C$, and each tuple in  $A \cup B$ is incident to exactly $l$ buckets in $C$. With the same argument, it takes $O(nkl)$ time for construct the tripartite graph representation. Moreover, it takes $O(\sum_\square |A_\square| \cdot |B_\square|)$ time for computing the quantity $\square_\OUT$ for all buckets, which can be further bounded by
\[\sum_{\square} |A_\square| \cdot |B_\square| < n \cdot \max_\square (|A_\square| + |B_\square|) = O(n \log n)\]
implied by Lemma~\ref{lem:ball-into-bin}. 

Consider any bucket $\square$ from hash function $g_j$. If the algorithm visits it during the enumeration, at least one pair of points within distance $r$ will be emitted, which has not been emitted by any bucket from hash function $h_i$ for $i < j$. Checking all pairs of points inside any bucket takes at most $O((d + kl) \cdot \max_{\square} |A_\square| \cdot |B_\square|)$ time, where it takes $O(d)$ time to compute the distance between any pair of points and $kl$ time for checking whether this pair has been emitted before or marking buckets which also witnesses this pair later. 
Thus, the delay between any two consecutive pairs of results is bounded by $O((d + kl) \cdot \max_{\square} |A_\square| \cdot |B_\square|)$, which is $\O(d)$ under the uniform assumption.  

Moreover, for each pair of points within distance $r$, it will be reported by any hash function with probability at least $p_1^k$. The probability that they do not collide on any one of hash function is at most $(1-p_1^k)^{3 \cdot 1/p_1^k \cdot \ln n} \le 1/n^3$. As there are at most $n^2$ such pairs of tuples, the probability that any one of them is not reported by our data structure is at most $1/n$. By a union bound, the probability that uniform assumption fails or one join result is not reported is at most $\frac{1}{n} + \frac{1}{n} = \frac{2}{n}$. Thus, the result holds with probability at least $1 - \frac{2}{n}$. 
\end{proof}


\subsection{Without Uniform Assumption}

In general, without this uniform assumption, we need to explore more properties of the LSH family for an efficient data structure. Our key insight is that after checking some pairs of points in one bucket (the specific numbers of pairs will be determined later), we can safely skip the bucket, since with high probability any join result missed in this bucket will be found in another one. In this way, we avoid spending too much time in one bucket before finding any join result.
%
Given a set $P$ of points and a distance threshold $r$, let $\overline{\mathcal{B}}(q, P, r) = \{p \in P\mid \dist(p,q) >r\}$. 
%
%
The next lemma follows from~\cite{indyk1998approximate, har2011geometric}.



\begin{lemma}
\label{lem:nnq}
    For a set $P$ of $n$ points in hamming space $\Re^d$ and a distance threshold $r$, if $k=O(\log n)$ and $\tau=O(n^{\rho})$, then for any point
    $p\in P$ the following conditions hold with constant probability $\gamma$:  for any $q\in P$ such that $\dist(p,q)\leq r$, there exists a bucket $\square $such that $p, q$ collide and $|\square \cap \overline{\mathcal{B}}(p,P,(1+\eps)r)| \le M$ for  $M=O(n^{\rho})$.
\end{lemma}

\subsubsection{Data structure}

We adopt a LSH family $\H$ with quality parameter $\rho$ and randomly choose $\tau = O(n^{\rho})$ hash functions $g_1, g_2,\cdots, g_\tau$. To ensure our high-probability guarantee (as shown later), we maintain $O(\log n)$ copies of this data structure.
We construct $m = \frac{3}{\log(1/\gamma)} \log n=
O(\log n)$ copies of the data structure as $\mathbb{I}_1, \mathbb{I}_2, \cdots, \mathbb{I}_m$. 

For each bucket $\square$, we store and maintain a set of $M$ arbitrary points $\AM\subseteq A_\square$ and $\BM\subseteq B_\square$. For each point $a\in \AM$ we maintain a counter $a_c=|\{b\in \BM\mid \dist(a,b)\leq 2(1+\eps)r\}|$.
A bucket $\square$ is \emph{active} if there exists a pair $(a,b)\in \AM\times \BM$ such that $\dist(a,b)\leq 2(1+\eps)r$. Equivalently, a bucket $\square$ is active if there exists $a\in \AM$ with $a_c>0$. All active pairs are maintained in a list $\C$. For each bucket $\square\in \C$ we store a \emph{representative} pair $(a_\square,b_\square)\in \AM\times \BM$ such that $\dist(a_\square, b_\square)\leq 2(1+\eps)r$.

For any pair of points $(a,b) \in A \times B$ and a hashing bucket $\square$, we refer to $\square$ as the {\em proxy bucket} for $(a,b)$ if (\romannumeral 1) $a \in A_\square, b \in B_\square$; (\romannumeral 2) $|\bar{\mathcal{B}}(a, A_\square \cup B_\square, (1+\eps)r)| \le M$.
Lemma~\ref{lem:non-uniform-property} implies that each join result $(a,b)$ with $\dist(a,b) \le r$ has at least one proxy bucket.
\begin{lemma}
    \label{lem:non-uniform-property}
    With probability at least $1-1/n$, for any pair of points $(a, b) \in A\times B$ with $\dist(a,b) \le r$, there exists a data structure $\mathbb{I}_j$ that contains a bucket $\square$ such that:
    \begin{itemize}
        \item $a,b$ will collide in $\square$;
        \item $|\square \cap \overline{\mathcal{B}}(a, A, (1+\epsilon)r)| \le M$ and $|\square \cap \overline{\mathcal{B}}(a, B, (1+\epsilon)r)| \le M$.
    \end{itemize}
\end{lemma}

\begin{proof}[Proof of Lemma~\ref{lem:non-uniform-property}]
    Consider any pair of points $(a\in A, b\in B)$ within distance $r$ and an arbitrary data structure constructed as described above. From Lemma~\ref{lem:nnq}, with probability at least $\gamma$ there exists a bucket in the data structure that contains both $a, b$ and the number of collisions of $a$ (with the rest of the points in $A\cup B$) is bounded by $M$.
    
    Let $F_j$ be the event that is true if there is a bucket in $\mathbb{I}_j$ that witnesses the collision of $a, b$ and the number of collisions of $a$ is bounded by $M$. Since $F_i, F_j$ are independent for $i\neq j$, we have $\Pr[\bar{F}_1 \cap \ldots \cap \bar{F}_C]=\Pr[\bar{F}_1]\cdot \ldots \cdot \Pr[\bar{F}_C] \leq \gamma^{\frac{3}{\log(1/\gamma)}\log n}\leq 1/n^3$. Let $Z$ be the number of pairs with distance at most $r$. We have $Z\leq n^2$. Let $G_i$ be the event which is true if for the $j$-th pair of points $a', b'$ with distance at most $r$, there is at least a copy of the data structure such that there exists a bucket that contains both $a', b'$ and the number of collisions of $a'$ is bounded by $M$. Then $\Pr[G_1 \cap \ldots \cap G_Z]=1-\Pr[\bar{G}_1\cup \ldots \cup \bar{G}_Z]\geq 1-\Pr[\bar{G}_1]-\ldots-\Pr[\bar{G}_Z]\geq 1-n^2/n^3\geq 1-1/n$. Hence, with high probability, for any pair $a\in A, b\in B$ with distance at most $r$ there will be at least one bucket in the data structure such that, both $a, b$ are contained in the bucket and the number of collisions of $a$ in the bucket is bounded by $M$.
\end{proof}

\begin{lemma}
\label{lem:proxy}
   For any bucket $\square$, if there exist $M$ points from $A_\square$ and $M$ points from $B_\square$, such that none of the $M^2$ pairs has its distance within $2(1+\eps)r$, $\square$ is not a proxy bucket for any pair $(a, b) \in A_\square \times B_\square$ with $\dist(a,b) \le r$. 
\end{lemma}
\begin{proof}
Let $A',B'$ be two sets of $M$ points from $A_\square, B_\square$ respectively. We assume that all pairs of points in $A' \times B'$ have their distances larger than $2(1+\eps)r$.
Observe that $\square$ is not a proxy bucket for any pair $(a\in A', b\in B')$.
It remains to show for $(a\in A_\square\setminus A', b\in B_\square)$ with $\dist(a,b) \le r$.
%
%
First, assume $b\in B_\square \setminus B'$.
If $A' \subseteq \bar{\mathcal{B}}(a,A,(1+\eps)r)$ or $B' \subseteq \bar{\mathcal{B}}(a,B,(1+\eps)r)$,
$\square$ is not a proxy bucket for $(a,b)$. Otherwise, there must exist at least one point $a' \in A'$ as well as $b' \in B'$ such that $\dist(a,a') \le (1+\eps)r$ and $\dist(a,b') \le (1+\eps)r$, so from triangle inequality $\dist(a',b') \le \dist(a,a') + \dist(a,b') \le 2(1+\eps)r$.
Thus, $(a', b') \in A' \times B'$ is a pair within distance $2(1+\eps)r$, coming to a contradiction. 
Then we assume that $b\in B'$.
If $A' \subseteq \bar{\mathcal{B}}(a,A,(1+\eps)r)$,
$\square$ is not a proxy bucket for $(a,b)$. Otherwise, there must exist at least one point $a' \in A'$ such that $\dist(a,a') \le (1+\eps)r$, so from triangle inequality $\dist(a',b) \le \dist(a,a') + \dist(a,b) \le (2+\eps)r\leq 2(1+\eps)r$.
Thus, $(a', b) \in A' \times B'$ is a pair within distance $2(1+\eps)r$, coming to a contradiction. 
\end{proof}

Later, we see that our enumeration phase only reports each join result in one of its proxy buckets. This guarantees the completeness of query results, but de-duplication is still necessary if a pair of points has more than one proxy buckets.

\subsubsection{Update} 
We handle insertion and deletion separately.
We assume that we insert or delete a point $a\in A$. We can handle an update from $B$, similarly. All pseudocodes are given below.

\begin{algorithm}
\caption{{\sc Insert}$(a \in A)$}
\ForEach{hash function $g$ in the data structure}{
    $\square \gets$ the bucket with hash value $g(a)$\;
    Insert $a$ into $A_\square$\;
    \If{$|\AM|<M$}{
        Insert $a$ into $\AM$\;
        Compute $a_c$ by computing $\dist(a,b)$ for each $b\in \BM$\;
        \If{$a_c>0$ AND $\square\notin \C$}{
            $\C\leftarrow \C\cup\{\square\}$\;
            $(a_\square, b_\square)=(a,b)$ for a point $b\in \BM$ with $\dist(a,b)\leq 2(1+\eps)r$\;
        }
    }
}
\end{algorithm}

\begin{algorithm}
\caption{{\sc Delete}$(a \in A)$}
\ForEach{hash function $g$ in the data structure}{
    $\square \gets$ the bucket with hash value $g(a)$\;
    Delete $a$ from $\square_A$\;
    \If{$a\in \AM$}{
        Remove $a$ from $\AM$\;
        Insert an arbitrary point $a'\in A_\square\setminus \AM$ into $\AM$\;
        \If{$\square\notin \C$ AND $a_c'>0$}{
            $\C\gets \C\cup \{\square\}$\;
            $(a_\square,b_\square)=(a', b)$ where $b\in \BM$ and $\dist(a',b)\leq 2(1+\eps)r$\;
        }
        \ElseIf{$\square\in \C$ AND $a_\square=a$}{
            \If{$\exists a''\in \AM$ with $a_c''>0$}{
                $(a_\square,b_\square)=(a'', b)$ where $b\in \BM$ and $\dist(a'',b)\leq 2(1+\eps)r$\;
            }\Else{
                $\C\gets \C\setminus \{\square\}$
            }
        }
    }
    }   
\end{algorithm}

\paragraph{Insertion of $a$.} We compute $g(a)$ for each chosen hash function $g$. Assume $\square$ is the bucket with hash value $g(a)$. We first insert $a$ to $A_\square$.
If $|\AM|<M$ we add $a$ in $\AM$ and we compute the counter $a_c$ by visiting all points in $\BM$.
If $\square$ was inactive and $a_c>0$, we add $\square$ in $\C$, we find a point $b\in \BM$ with $\dist(a,b)\leq 2(1+\eps)r$ and we set the representative pair $(a_{\square}, b_\square)=(a,b)$. If $|\AM|=M$ before we insert $a$ we do not do anything.

 \paragraph{Deletion of $a$.} Similarly, we compute $g(a)$ for each chosen hash function $g$. Assume $\square$ is the bucket with hash value $g(a)$. We first remove $a$ from $A_{\square}$.
 If $a\in \AM$ we also remove it from $\AM$ and we replace it with an arbitrary point $a'\in A_{\square}\setminus \AM$ by computing its counter $a_c'$.
 If $a$ was a point in the representative pair of $\square$ we update it by finding any point $a''\in \AM$ with $a_c''>0$. If there is not such point we remove $\square$ from $\C$.


\medskip
When there are $n/2$ updates, we just reconstruct the entire data structure from scratch.

\subsubsection{Enumeration}
The high-level idea is to enumerate the representative pair of points for each bucket in $\C$. Assume a representative pair $(a,b)$ is found in a bucket $\square \in \C$. Next, the algorithm is going to enumerate all pairs containing point $a$. 

Initially, all buckets containing $a$ are maintained in $\C(a) \subseteq \C$. 
Algorithm~\ref{alg:enumerate-non-uniform} visits every bucket $\square \in \C(a)$ and starts to check the distances between $a$ and points in $B_\square$ that are not marked by $X(\square, a)$ (we show when a point is marked in the next paragraph). Each time a pair $(a,b)$ within distance $2(1+\eps)r$ is found, it just reports this pair and calls the procedure {\sc Deduplicate} on $(a,b)$ (details will be given later). If there are more than $M$ points far away (i.e. $>2r(1+\eps)$) from $a$, we just stop enumerating results with point $a$ in this bucket, and remove the bucket\footnote{In the enumeration phase, the ``remove'' always means conceptually marked, instead of changing the data structure itself.} $\square$ from $\C(a)$.
We also update $\AM$ so that if $a\in \AM$ we replace $a$ with another point $A_\square$. 
Once the enumeration is finished on $a$, 
i.e., when $\C(a)$ becomes empty, it can be easily checked that $a$ has been removed from all buckets.

Next, we explain more details on the de-duplication step presented as Algorithm~\ref{alg:deduplicate}. Once a pair of points $(a,b)$ within distance $2(1+\eps)r$ is reported, Algorithm~\ref{alg:deduplicate} goes over all buckets witnessing the collision of $a,b$, and marks $b$ with $X(\square, a)$ to avoid repeated enumeration (line 2). Moreover, for any bucket $\square$ with $a \in A_\square$ and $b \in B_\square$, if $(a,b)$ is also its representative pair, Algorithm~\ref{alg:deduplicate} performs more updates for $\square$. Algorithm~\ref{alg:deduplicate} first needs to decide whether $\square$ is still an active bucket for $a$ by checking the distances between $a$ and $M$ points unmarked by $a$ in $B_\square$. If such a pair within distance $2(1+\eps)r$ is found, it will set this pair as new representative for $\square$. Otherwise, it is safe to skip all results with point $a$ in this bucket. In this case, it needs to further update a new representative pair for $\square$ using $\AM, \BM$. Moreover, if no representative pair can be found, it is safe to skip all results with bucket $\square$.

\begin{algorithm}
\caption{{\sc EnumerateLSH}}
\label{alg:enumerate-non-uniform}

\While{$\C \neq \emptyset$}{
    $(a,b) \gets$ the representative pair of any bucket in $\C$\;
    $\C(a) \gets \{\square \in \C: a \in A_\square\}$\;
    \While{$\C(a) \neq \emptyset$}{
        Pick one bucket $\square \in \C(a)$; $i \gets 0$\;
        \ForEach{$b \in B_\square - X(\square, a)$}{
            \If{$\dist(a,b) \le 2(1+\eps)r$}{
                {\sc Emit}$(a,b)$\;
                
                {\sc Deduplicate}$(a,b)$\;
            }
            \Else{
                $i \gets i + 1$\;
                \bf{if} {$i > M$} \bf{then break\;} 
            }
        }
        $A_\square \gets A_\square - \{a\}$\; 
        $\C(a) \gets \C(a) - \{\square\}$\;
        Replace $a$ from $\AM$ (if $a\in \AM$) and update its representative pair\;
    }
}
\end{algorithm}

\begin{algorithm}
\caption{{\sc Deduplicate$(a,b)$}}
\label{alg:deduplicate}

\ForEach{$\square \in C$ with $a \in A_\square$ and $b \in B_\square$ }{
        $X(\square, a) \gets X(\square, a) \cup \{b\}$\;
        \If{$(a_\square,b_\square)=(a,b)$}{
                $B' \gets M$ arbitrary points in $B_\square - X(\square, a)$\; 
               \If{there is $b' \in B'$ with $\dist(a, b') \le 2(1+\eps)r$}{
                    $(a_\square,b_\square)=(a, b')$\;
                }
                \Else{
                    $\C(a) \gets \C(a) - \{\square\}$\;
                   $A_\square \gets A_\square - \{a\}$\;
                   \If{$a\in \AM$}{
                        Replace it with a new item $a'\in A_{\square}\setminus \AM$\;
                        Compute $a_c'$\;
                        \If{$\exists a''\in \AM$ with $a_c''>0$}{
                            $(a_{\square},b_{\square})=(a'', b'')$ where $b''\in \BM$ and $\dist(a'',b'')\leq 2(1+\eps)r$\;
                        }\Else{
                            $\C\gets \C\setminus\{\square\}$\;
                        }
                   }
        
                }
            }
    }
\end{algorithm}

For any bucket $\square$, we can maintain the points in $A_\square, B_\square, X(\square, a)$ in balanced binary search trees, so that points in any set can be listed or moved to a different set with $O(\log n)$ delay.
Moreover, to avoid conflicts with the markers made by different enumeration queries, we generate them randomly and delete old values by lazy updates~\cite{erickson2011static, overmars1981worst, overmars1987design} after finding new pairs to report.
%
%
%
%
%

\begin{lemma}
\label{lem:non-uniform-approximate}
The data structure supports $(1+2\eps)$-approximate enumeration, with high probability.
\end{lemma}

\begin{proof}[Proof of Lemma~\ref{lem:non-uniform-approximate}]
    It can be easily checked that any pair of points with distance more than $2(1+\eps)r$ will not be enumerated. Also, each result is reported at most once by Algorithm~\ref{alg:deduplicate}. Next, we will show that with high probability, all pairs of points within distance $r$ are reported. Consider any pair of points $(a,b)$ within distance $r$. Implied by Lemma~\ref{lem:non-uniform-property}, there must exist a proxy bucket $\square$ for $(a,b)$. Observe that there exists no subset of $M$ points from $A_\square$ as $\AM$ and subset of $M$ points from $B_\square$ as $\BM$, where all pairs of points in $\AM \times \BM$ have their distances larger than $2(1+\eps)r$, implied by Lemma~\ref{lem:proxy}, so $\square$ is active. Moreover, there exists no subset of $M$ points from $B_\square$ as $\BM$, where all pairs of points $(a, b'\in \BM)$ have their distances larger than $2(1+\eps)r$, so $\square$ is an active bucket for $a$. In Algorithm~\ref{alg:enumerate-non-uniform}, when visiting $\square$ by line 7-18, $(a,b)$ must be reported by $\square$ or have been reported previously.
\end{proof}

We next analyze the complexity of the data structure. The size of the data structure is $\O(dn + n  k \tau)=\O(dn+n^{1+\rho})$. The insertion time is $\O(d \tau  M)=\O(dn^{\rho}M)$. Using that, we can bound the construction time of this data structure as $\O(d n\tau  M)=\O(dn^{1+\rho}M)$.
The deletion time is $\O(d  \tau  M)=\O(d n^{\rho}M)$. The delay is bounded by $\O(d \tau M)=\O(d n^{\rho}M^2)$ since after reporting a pair $(a, b)$, we may visit $\O(\tau)$ buckets and spend $O(M)$ time for each in updating the representative pair.
%
%
%
%
Putting everything together, we conclude the next theorem.

\begin{theorem}
\label{the:high-non-uniform}
Let $A$ and $B$ be two sets of points in $\mathbb{R}^d$, where $|A|+|B|=n$ and let $\eps, r$ be positive parameters. For $\rho=\frac{1}{(1+\eps)^2}+o(1)$,
a data structure of $\O(dn+n^{1+\rho})$ size can be constructed in $\O(dn^{1+2\rho})$ time, and updated in $\O(dn^{2\rho})$ amortized time, while supporting $(1+2\eps)$-approximate enumeration for similarity join under the $\ell_2$ metric with $\O(dn^{2\rho})$ delay.
\end{theorem}

Alternatively, we can insert or delete points from $A\cup B$ without maintaining the sets $\AM, \BM$ for every bucket $\square$. In the enumeration phase, given a bucket $\square$, we can visit $M$ arbitrary points from $A_\square$ and $M$ arbitrary points from $B_\square$ and compute their pairwise distances. If we find no pair $(a\in A_\square, b\in B_\square)$ with $\dist(a,b)\leq 2(1+\eps)r$ then we skip this bucket. Otherwise we report the pair $(a,b)$ and we run the deduplicate procedure. In this case the update time is $\O(dn^{\rho})$ and the delay is $\O(dn^{3\rho})$.

\begin{theorem}
\label{the:high-non-uniform2}
Let $A$ and $B$ be two sets of points in $\mathbb{R}^d$, where $|A|+|B|=n$ and let $\eps, r$ be positive parameters. For $\rho=\frac{1}{(1+\eps)^2}+o(1)$,
a data structure of $\O(dn+n^{1+\rho})$ size can be constructed in $\O(dn^{1+\rho})$ time, and updated in $\O(dn^{\rho})$ amortized time, while supporting $(1+2\eps)$-approximate enumeration for similarity join under the $\ell_2$ metric with $\O(dn^{3\rho})$ delay.
\end{theorem}

Notice that the complexities of the theorems above depend on the parameter $M$ from Lemma~\ref{lem:nnq}. Hence, a better bound on $M$ will give improve the results of our data structure. In the original paper \cite{indyk1998approximate} (Section 4.2) for the Hamming metric the authors choose $\rho=\frac{\log(1/p_1)}{\log(p_1/p_2)}$ showing that for any $p,q\in P$ such that $\dist(p,q)\leq r$ there exists a bucket, with constant probability $\gamma$, that $p,q$ collide and the number of points in $P\cap \overline{\mathcal{B}}(p,(1+\eps)r)$ colliding with $p$ in the bucket is at most $M=O(1)$.
For $\eps>1$ they show that $\rho<\frac{1}{\eps}$.
Equivalently we can set $\eps$ as $\eps-1$ and 
$M=O(1)$.
Using this result we can get the next theorem.

\begin{theorem}
\label{the:Shigh-non-uniform}
Let $A,B$ be two sets of points in $\mathbb{H}^d$, where $|A|+B|=n$ and let $\eps, r$ be positive parameters. For $\rho=\frac{1}{1+\eps}$,
a data structure of $\O(dn+n^{1+\rho})$ size can be constructed in $\O(dn^{1+\rho})$ time, and updated in $\O(dn^{\rho})$ amortized time, while supporting $(3+2\eps)$-approximate enumeration for similarity join under the Hamming metric with $\O(dn^{\rho})$ delay.
\end{theorem}

In the next remarks we show that our results can be extended to $r$ being part of the query (variable).
Furthermore, we show that our result is near-optimal.

\paragraph{Remark 1.} Similar to the LSH~\cite{indyk1998approximate} used for ANN query, we can extend our current data structure to the case where $r$ is also part of the query. For simplicity, we focus on Hamming metric. For $\mathbb{H}^d$, it holds that $1\leq r\leq d$. Hence, we build $Z=O(\log_{1+\eps} d)=O(\eps^{-1}\log d)$ data structures as described above, each of them corresponding to a similarity threshold $r_i=(1+\eps)^i$ for $i=1,\ldots, Z$. Given a query with threshold $r$, we first run a binary search and find $r_j$ such that $r\leq r_j\leq (1+\eps)r$. Then, we use the $j$-th data structure to answer the similarity join query. Overall, the data structure has $\O(dn+\eps^{-1}n^{1+\rho}\log d)$ size can be constructed in $\O(\eps^{-1}dn^{1+\rho}\log d)$ time, and updated in $\O(\eps^{-1}dn^{\rho}\log d)$ amortized time.
After finding the value $r_j$ in $O(\log (\eps^{-1}\log d))$ time, the delay guarantee remains $\O(dn^{\rho})$.
We can also extend this result to $\ell_2$ or $\ell_\infty$ metrics using known results (\cite{gionis1999similarity, har2011geometric, indyk1998approximate}).



\paragraph{Remark 2.} It is known that the algorithm for similarity join can be used to answer the ANN query. Let $P$ be a set of points in $\Re^d$, where $d$ is a large number, and $\eps, r$ be parameters. The ANN query asks that (1) if there exists a point within distance $r$ from $q$, any one of them should be returned with high probability; (2) if there is no point within distance $(1+\eps)r$ from $q$, it returns ``no''with high probability.  For any instance of ANN query, we can construct an instance of similarity join by setting $A  = P$ and $B = \emptyset$. Whenever a query point $q$ is issued for ANN problem, we insert $q$ into $B$, invoke the enumeration query until the first result is returned (if there is any), and then remove $q$ from $B$. 
Our data structure of $\O(dn + n^{1+\rho})$ size can answer $(1+2\eps)$-approximate ANN query in $\O(dn^{2\rho})$ time in $\ell_2$, which is only worse by a factor $n^\rho$ from the best data structure for answering $\eps$-approximate ANN query.